\pdfoutput=1
\documentclass[journal]{IEEEtran}

\usepackage{multirow,balance}
\usepackage{bm}
\usepackage{relsize}
\usepackage{graphicx,cite}
\usepackage[T1]{fontenc}

\usepackage{comment}
\usepackage{cite}
\usepackage{amsmath, amssymb,amsthm}

\newtheorem{definition}{Definition}
\newtheorem{theorem}{Theorem}
\newtheorem{proposition}{Proposition}
\newtheorem{lemma}{Lemma}
\newtheorem{corollary}{Corollary}
\newtheorem{remark}{Remark}

\usepackage{algorithm}
\usepackage{algorithmic}
\usepackage{array}
\usepackage{longtable,enumitem}

\usepackage{listings}
\usepackage{color}

\usepackage[caption=true,font=footnotesize]{subfig}

\usepackage[hyphens]{url}
\usepackage{hyperref}
\hypersetup{colorlinks={true},linkcolor={blue},citecolor=blue}

\usepackage{cleveref}
\crefname{section}{}{\S\S}
\crefdefaultlabelformat{\S#2#1#3}

\newcommand{\gt}[1]{{\color{blue} #1}}

\begin{document}
%
 \title{Node-Constrained Traffic Engineering:\\Theory and Applications}

%
%

\author{George~Trimponias,
	Yan~Xiao,
	Xiaorui~Wu,
	Hong~Xu,
	and~Yanhui~Geng
	 \thanks{{\color{magenta}This manuscript is the extended version of the
	 IEEE/ACM ToN paper ``Node-Constrained Traffic Engineering: Theory and Applications'' with the addition of Remarks \ref{remark-1}, \ref{remark-2}, \ref{remark-3} in \cref{sec:directed}, Remark \ref{remark:undirected} in \cref{sec:undirected}, and Remark \ref{remark-4} in \cref{sec:group-centrality}.}}
	 \thanks{The work was supported in part by contract research between City
	 University of Hong Kong and Huawei (project no. 9231208) and a CRF grant
	 from the Research Grants Council of the HKSAR, China (C7036-15G). The
	 corresponding author is Hong Xu.}
	 \thanks{G. Trimponias is with Huawei Noah's Ark Lab, Hong Kong (email: \href{mailto:g.trimponias@huawei.com}{g.trimponias@huawei.com}). }
	 \thanks{Y. Xiao,  X. Wu and H. Xu are with Department of Computer Science, City University of Hong Kong, Hong Kong (email: \href{mailto:yanxiao6-c@my.cityu.edu.hk}{yanxiao6-c@my.cityu.edu.hk}, \href{mailto:xiaorui.wu@my.cityu.edu.hk}{xiaorui.wu@my.cityu.edu.hk}, \href{mailto:henry.xu@cityu.edu.hk}{henry.xu@cityu.edu.hk}).}
	  \thanks{Y. Geng is with Huawei Montreal Research Center, Canada (email: \href{mailto:geng.yanhui@huawei.com}{geng.yanhui@huawei.com}). }
}
%
%

\markboth{IEEE/ACM Transactions on Networking, 2019}%
{Trimponias \MakeLowercase{\textit{et al.}}: Node-Constrained Traffic Engineering: Theory and Applications}
%



\maketitle


\begin{abstract}
Traffic engineering (TE) is a fundamental task in networking. 
Conventionally, traffic can take any path connecting the source and destination. 
Emerging technologies such as segment routing, however, use logical paths going through a predetermined set of middlepoints. 
Inspired by this, in this work we introduce the problem of node-constrained TE, where traffic must go through a set of middlepoints, and study its theoretical fundamentals. 
We show that the general node-constrained TE that constrains the traffic to take paths going through one or more middlepoints is NP-hard for directed graphs but strongly polynomial for undirected graphs, unveiling a profound dichotomy between the two cases. 
We additionally investigate the popular variant of node-constrained TE that uses shortest paths between middlepoints, and show that the problem can now be solved in weakly polynomial time for a fixed number of middlepoints. 
Yet if we constrain the end-to-end paths to be acyclic, the problem can become NP-hard.
This explains why existing work focuses on the computationally tractable variant. 
An important application of our work concerns the computational complexity of flow centrality, first proposed in 1991 by Freeman et al. \cite{fbw1991}: we show that it is NP-hard for directed but strongly polynomial for undirected graphs.
Finally, we investigate the middlepoint selection problem in general node-constrained TE. 
We 
introduce group flow centrality as a solution concept for multi-commodity networks, study its complexity, and show that it is monotone but not submodular for both directed and undirected graphs.
Our work provides a thorough theoretical treatment of node-constrained TE and its applications.

\end{abstract}


\section{Introduction}
\label{sec:intro}

Traffic engineering (TE) is an important task for network operators to improve network efficiency and application performance. TE is exercised in a wide range of networks, from carrier networks \cite{HVSB15,FT00} to data center backbones \cite{JKMO13,HKMZ13}. Increasingly, TE is implemented using SDN (Software Defined Networking) given its flexibility. Notable examples include Google's B4 \cite{JKMO13} and Microsoft's SWAN \cite{HKMZ13}. 
Implementing TE in the data plane requires a large number of flow table entries on switches. This is because each switch on the path needs to have an entry per demand, i.e. ingress-egress switch pair, to forward its traffic to the next hop, and for a large-scale network there can be many demands \cite{JKMO13,HKMZ13}. 

Segment routing \cite{segment1, segment2, segment3} is a recently proposed routing architecture to tackle this challenge. 
Its key idea is to perform routing based on a sequence of logical segments formed by a set of {\em middlepoints}\footnote{This term is in accordance with the prior literature on segment routing.} between the ingress and egress nodes. A segment is the logical pipe between two middlepoints that may include multiple physical paths spanning multiple hops, and static hashing is used to load balance traffic among these paths. To simplify, usually only shortest paths are used between two middlepoints.
Now with segment routing, instead of end-to-end paths, intermediate switches only need to know how to reach middlepoints in order to forward packets. 
This can greatly reduce the overhead and cost of TE \cite{bhatia2015optimized,HVSB15}. 

TE with segment routing is different from traditional TE, where the traffic from a source to a destination can use any path. This motivates us to introduce the class of \textit{node-constrained} TE, which includes any TE variant where the traffic is constrained to go through one or more predetermined middlepoints. Segment routing corresponds to a specific variant of node-constrained TE, which only uses shortest paths between the middlepoints. 

The general node-constrained routing is important for TE.
First, like segment routing, the use of middlepoints saves precious flow table resources and reduces the overhead of implementing TE with finer granularity flow control. 
Second, the use of shortest paths in segment routing may limit the TE performance and robustness. 
Some shortest paths may involve the same link which degrades the throughput one can use effectively. 
Perhaps more importantly, operators prefer edge-disjoint paths over shortest paths for better diversity and robustness in cases of link failures \cite{LKMZ14}, where failover can be done by routing through at least the remaining paths. 
Last but not least, the broader and fundamental node-constrained TE problem has not received much attention in the networking community, despite its significant application potential.

Inspired by the above, in this work we investigate the theoretical fundamentals of node-constrained TE. 
We consider two common types of TE introduced in \cref{sec:TE} depending on the objective: $TE_{MF}$ maximizes the total throughput based on multi-commodity flow, while $TE_{LU}$ minimizes the maximum link utilization. 
The two types are closely related.

Our analysis is organized in three parts.
We start in \cref{sec:theory} with the most general node-constrained TE problem, where traffic can take any path as long as that path goes through a set of middlepoints.
For directed graphs, we prove that the decision version of $TE_{MF}$ is NP-hard, even with just a single middlepoint. 
Due to the connection between the decision version of maximum flow and $TE_{LU}$, this implies that $TE_{LU}$ is NP-hard too. 
For undirected graphs, we show that TE is strongly polynomial, since it can be equivalently written as a special linear program via a polynomial transformation. 
Therefore, we establish that node-constrained TE is NP-hard and difficult to solve optimally in general, as most TE problems use directed graphs to model bidirectional links and traffic.

Given the hardness results, we next investigate in \cref{sec:opt} a variant of node-constrained TE with shortest paths,  where the traffic uses only shortest paths between any two middlepoints. 
We wish to see if this variant,  inspired by segment routing, makes the TE problem easier to solve. 
We prove that both $TE_{MF}$ and $TE_{LU}$ can now be solved in weakly polynomial time by transforming them into linear programs, when the number of middlepoints per path is fixed. 
Our results thus provide a theoretical foundation for existing work that focuses on shortest path based segment routing \cite{bhatia2015optimized,HVSB15} and not the more general variant. 
We further note that this variant may end up with end-to-end paths that contain cycles, since the various segments may repeat the same edge. 
Cyclic paths are clearly bad for TE as the precious WAN bandwidth is wasted sending traffic back and forth. 
For this reason, we study a different variant that requires acyclic end-to-end paths, which is a specific case of node-constrained TE with shortest paths. 
We show that imposing this constraint generally renders TE NP-hard again. 

Lastly, we study another fundamental and practical question in node-constrained TE: how to select the middlepoints that yield good TE performance?
We investigate \textit{flow centrality} as a potential solution approach to this problem in \cref{sec:flowcentrality}.
Flow centrality, first introduced in 1991 \cite{fbw1991}, determines how important a node is in terms of the percentage of the maximum flow that can go through that node over all possible demands. 
It serves as a natural criterion for middlepoint selection in the general node-constrained TE, i.e. we can select as middlepoints the top-$k$ nodes with the highest flow centrality. 

Our analysis implies that the flow centrality is NP-hard to compute in directed graphs unlike other common centrality concepts (\cref{sec:related}), but strongly polynomial in undirected graphs. 
Furthermore, since flow centrality only concerns individual nodes, we propose {\em group flow centrality}, which generalizes flow centrality to a group of nodes in order to better solve the middlepoint selection problem.
We introduce the related concept of $N$-group maximum flow, which corresponds to the problem of determining a set of middlepoints that maximizes the amount of flow that can go through any node in the set.
We show it is NP-hard; furthermore, unlike other common group graph centralities it is monotone but not submodular in both directed or undirected graphs, which implies that the standard greedy algorithm \cite{nwf1978} with $(1-\frac{1}{e})$-approximation ratio is not applicable.


We make several contributions in this paper. 
\begin{itemize}
	\item 
	We provide the first systematic study of node-constrained TE, which lays down the groundwork for understanding its theoretical fundamentals. 
	Our analysis shows that node-constrained TE is generally NP-hard hard for directed graphs, except for the variant that only uses shortest paths between middlepoints. 
	Our study further touches on the middlepoint problem in node-constrained TE.
	We study flow centrality and propose group flow centrality to select the best set of $N$ nodes that maximizes the total flow, and analyze the computational complexity.

	\item
	Our theoretical results shed light on the development of the emerging node-constrained TE in practice.
	Our hardness results indicate that efficient heuristics or approximation algorithms are in urgent need in several cases, including general node-constrained TE problems with directed graphs which most TE problems use, and variants of node-constrained TE with acyclic paths. 
	Middlepoint selection is another promising area for future work.
	More study is needed to make the graph theoretical approach feasible, for both flow centrality 
for individual nodes and group flow centrality which has not been well explored.

	\item
	Finally, some of our results are interesting in their own right in the corresponding theoretical contexts.
	For example, in \cref{sec:theory} we unveil a dichotomy between the directed and undirected cases in terms of the complexity of node-constrained TE.
	A similar dichotomy is found in \cref{sec:flowcentrality}, where we prove that flow centrality, a previously introduced but little understood graph centrality concept, is NP-hard to compute in directed graphs but strongly polynomial in undirected graphs.
\end{itemize}

\section{Background on Traffic Engineering}
\label{sec:TE}

We first introduce some background on traffic engineering (TE) in this section.
In our work, we focus on two common types of TE depending on the objective criterion. The first maximizes the total throughput subject to the capacity and maximum demand constraints. Since it can be formulated as a maximum flow problem, we call it $TE_{MF}$. The second type minimizes the maximum link utilization, which acts as the system bottleneck. For this reason, we call it $TE_{LU}$. 

The rest of this section is organized as follows. 
We introduce some preliminary notions and concepts in \cref{sec:TE-prelim}. 
We then present $TE_{MF}$ in \cref{sec:TE-MCF} and $TE_{LU}$ in \cref{sec:TE-LU}. 
Lastly we show an interesting connection between the decision version of $TE_{MF}$ and the optimal solution to $TE_{LU}$ in \cref{sec:connection}.

\subsection{Preliminaries}
\label{sec:TE-prelim}

Assume a directed graph $G=(V,E)$, where $V$ is the set of nodes and $E$ the set of directed edges. Given a node $v\in V$, $v^+$ denotes the set of outgoing edges of node $v$, i.e., the subset of edges in $E$ of the form $(v,u)$, $u\in V$. Similarly, the set $v^-$ denotes the set of incoming edges of $v$ of the form $(u,v)$, $u\in N$. The out-degree of $v$ is defined as the cardinality $|v^+|$, whereas the in-degree is defined as the cardinality $|v^-|$.

A \textit{flow network} $G=(V,E,c)$ is defined as a directed graph $G=(V,E)$, together with a non-negative function $c:V\times V\rightarrow\mathbb{R}_{\geq 0}$ that assigns to each edge $e\in E$ a non-negative capacity $c(e)$. If $(u,v)\not\in E$, then we define $c(u,v)=0$.

A \textit{walk} in a directed graph is an alternating sequence of vertices and edges, $v_0$, $e_0$, $v_1$, $\dots$, $v_{k-1}$, $e_{k-1}$, $v_k$, which begins and ends with vertices and has the property that each $e_i$ is an edge from $v_i$ to $v_{i+1}$. A \textit{path} is a walk where all edges are distinct. A \textit{simple path} is a path where all vertices are distinct. The term $u-v$ path (resp., simple path) refers to any valid path (resp., simple path) from $u$ to $v$.

In flow networks, we usually distinguish between single-commodity and multi-commodity flows. For single-commodity flow problems, we consider a single commodity\footnote{When it is clear from the context, we use the terms {\em commodities}, {\em demands}, and {\em flows} interchangeably.} that consists of a source $s\in V$ and a sink $t\in V$, where $s\neq t$. 
For multi-commodity flows, we assume $L$ commodities of the form $(s_i,t_i)$, where $s_i,t_i\in V, s_i\neq t_i$. Each commodity $i$ is associated with a non-negative demand $D_i\geq0$. For convenience, we also use the notation $\bm{s}=(s_1,\dots,s_L)$ and $\bm{t}=(t_1,\dots,t_L)$, and write $(\bm{s},\bm{t})$ to denote the corresponding multi-commodity network. 

\subsection{TE Type 1: $TE_{MF}$}
\label{sec:TE-MCF}

Let $\mathcal{P}_i$ be the set of all $s_i-t_i$ paths, and $\mathcal{P}_{i,e}$ the set of all $s_i-t_i$ paths that go through edge $e$. Then the maximum multi-commodity flow program can be expressed via the following path-based formulation:
\begin{align}
\text{maximize}\qquad & \nu=\sum_{i=1}^L\sum_{p\in \mathcal{P}_i}f_i(p)\label{TE-MCF-obj}\\
\text{subject to }\qquad & \sum_{i=1}^L\sum_{p\in \mathcal{P}_{i,e}} f_i(p)\leq c(e), \forall e\in E\label{TE-MCF-first-constraint}\\
& \sum_{p\in \mathcal{P}_i}f_i(p)\leq D_i\label{TE-MCF-second-constraint}\\
& f_i(p)\geq 0, \forall i\in\{1,\dots,L\}, \forall p\in \mathcal{P}_i\label{TE-MCF-last-constraint}
\end{align}

In the above formulation, we divide the total flow into $L$ subflows, one per commodity. 
The subflow $f_i$ along path $p\in\mathcal{P}_i$ for commodity $i$ is $f_i(p)$. Constraint \eqref{TE-MCF-first-constraint} is a capacity constraint that the sum of all subflows on any edge cannot exceed the edge capacity. Constraint \eqref{TE-MCF-second-constraint} describes the maximum demand $D_i$ for commodity $i$\footnote{When the maximum demand $D_i$ is infinite, the corresponding demand constraint \eqref{TE-MCF-second-constraint} is trivially satisfied and can thus be removed.}. Finally, constraint \eqref{TE-MCF-last-constraint} imposes that each subflow should be non-negative. For any valid flow $f$, the value of a flow $\nu(f)$ is defined as the total sum of units that all subflows $f_i$ send. The value of the maximum flow is denoted as $\nu_{max}$.
$TE_{MF}$ is mostly used in data center backbone WANs \cite{JKMO13,HKMZ13}, where traffic is elastic, the operator controls not only the links but also the demands of applications, and the main objective is to fully utilize the expensive WAN links.

Note that even though the single-commodity maximum flow accepts various combinatorial algorithms \cite{amo1993}, e.g., Ford-Fulkerson or Edmonds-Karp, there is to date no combinatorial algorithm for the maximum multi-commodity flow even though the problem is known to be strongly polynomial \cite{Tardos1986}. Furthermore, even though a single-commodity network with integral capacities always accepts an integral maximum flow, this is not always the case with multi-commodity networks; in fact, the decision problem of integral multi-commodity flow is NP-complete even if the number of commodities is two, for both directed and undirected networks \cite{eis1975}.


\subsection{TE Type 2: $TE_{LU}$}
\label{sec:TE-LU}

$TE_{LU}$ is mostly used in carrier networks \cite{HVSB15,FT00}, where traffic demands are exogenous and inelastic, and the main objective thus is to control the congestion or link utilization in order to ensure the smooth operation of the network. 
The general form for this type of TE is:
\begin{align}
\text{minimize}\qquad & \theta\label{TE-LU-obj}\\
\text{subject to }\qquad & \sum_{i=1}^L\sum_{p\in \mathcal{P}_{i,e}} f_i(p)\leq \theta\cdot c(e), \forall e\in E\label{TE-LU-first-constraint}\\
& \sum_{p\in \mathcal{P}_i}f_i(p)\geq D_i, \forall i\in\{1,\dots,L\}\label{TE-LU-second-constraint}\\
& f_i(p)\geq 0, \forall i\in\{1,\dots,L\}, \forall p\in \mathcal{P}_i\label{TE-LU-last-constraint}
\end{align}

The variable $\theta$ in objective \eqref{TE-LU-obj} refers to the maximum link utilization, which must be minimized. Constraint \eqref{TE-LU-first-constraint} ensures that $\theta$ is at least as large as the maximum link utilization; constraint \eqref{TE-LU-second-constraint} ensures that each demand is satisfied; and the last constraint \eqref{TE-LU-last-constraint} is similar to $TE_{MF}$ in \cref{sec:TE-MCF}.

\subsection{Relationship between $TE_{MF}$ and $TE_{LU}$}
\label{sec:connection}
A natural question is whether the two types of TE are related. To answer this question, we first introduce the \textit{decision version} of the maximum multi-commodity flow problem.

\begin{definition}\label{def:decision}[Decision version of maximum flow (DMF)]
Given a flow network $G=(V,E,c)$ with a set of $L$ commodities $(\bm{s},\bm{t})$, each associated with a non-negative maximum demand $D_i\geq0$, decide whether the maximum multi-commodity flow has a value of at least $\sum\limits_{i=1}^L{D_i}$.
\end{definition}

Note that if the answer to the decision problem DMF is a ``yes'', then by constraint \eqref{TE-MCF-second-constraint} the maximum flow has to be exactly equal to $\sum\limits_{i=1}^L{D_i}$. If the answer is no, then the maximum flow is strictly less than $\sum\limits_{i=1}^L{D_i}$.
Our next result establishes the relationship between the two types of TE:

\begin{lemma}\label{lemma:connection}
DMF accepts a ``yes'' answer, if and only if the system \eqref{TE-LU-obj}--\eqref{TE-LU-last-constraint} for $TE_{LU}$ accepts a solution $\theta^*\leq1$.
\end{lemma}
\begin{proof}
Assume that DMF accepts a ``yes'' answer. Then there is a flow that respects constraints \eqref{TE-MCF-first-constraint}--\eqref{TE-MCF-last-constraint}. That flow will then trivially satisfy constraints \eqref{TE-MCF-first-constraint}--\eqref{TE-MCF-last-constraint} with $\theta=1$. Since the objective criterion of $TE_{LU}$ minimizes over $\theta$, the optimal solution to the TE program \eqref{TE-LU-obj}--\eqref{TE-LU-last-constraint} will accept an optimal solution $\theta^*\leq 1$.
For the reverse direction, assume that the system \eqref{TE-LU-obj}--\eqref{TE-LU-last-constraint} accepts a solution $\theta^*\leq1$. Then constraint \eqref{TE-LU-first-constraint} implies that the capacity constraints are satisfied for each edge, thus the corresponding flow is a valid flow for system \eqref{TE-MCF-obj}--\eqref{TE-MCF-last-constraint} with value $\sum\limits_{i=1}^L{D_i}$. The maximum flow has then trivially a value of at least $\sum\limits_{i=1}^L{D_i}$.
\end{proof}

Lemma~\ref{lemma:connection} shows that solving $TE_{LU}$ immediately generates a ``yes'' or ``no'' answer to the DMF. Thus, the TE naturally encompasses the general DMF problem of Definition~\ref{def:decision}. This also suggests that hardness results on the DMF (Proposition~\ref{prop:NPhardness}) immediately imply hardness for $TE_{LU}$. 

We conclude this section with two observations. First, even though we assumed a directed network throughout this section, it is possible to extend the definitions to undirected graphs as well\footnote{In the former (resp., latter) case we refer to flow networks with directed (resp., undirected) edges.}. The main difference is that an undirected edge is associated with a capacity, and flow can travel in both directions of a link, under the constraint that the sum of the flow value in the two edge directions does not exceed the capacity. 
Second, it is simple to reduce edge-constrained TE to node-constrained TE in both directed and undirected graphs. Indeed, we can replace any (directed or undirected) edge $(u,v)$ by the two consecutive edges $(u,z)$ and $(z,v)$ by introducing a new node $z$. Then the traffic constrained to go through edge $e$ can be equivalently characterized as the traffic going through node $z$. Node-constrained TE is thus at least as hard as edge-constrained TE, which is why we focus on the former. 

\section{General Node-Constrained\\ Traffic Engineering}
\label{sec:theory}

In this section, we study the general (unrestricted) node-constrained TE problem, where traffic can take any path as long as that path goes through a set of middlepoints.
We focus on the simplest setting where a path goes through a specific node $w$, in order to establish the hardness results.
For directed graphs, we show in \cref{sec:directed} that the decision version of the maximum $w$-flow is NP-hard, which implies that $TE_{LU}$ is NP-hard as well. 
On the other hand, we show in \cref{sec:undirected} that in undirected graphs the problem is strongly polynomial after equivalently rewriting it as a special linear program. 
At the end of the section, we show that our results also extend to the general case where traffic goes through \textit{at least} one node from the $k>1$ middlepoints $w_1,\dots,w_k$, where $k$ is fixed and not part of the input.
The practical significance of our results is that we rigorously establish that node-constrained TE is NP-hard and difficult to solve in general, as most TE problems use directed graphs to model bidirectional links and traffic.


\subsection{The Directed Case}
\label{sec:directed}
\subsubsection{Hardness of $TE_{MF}$}

The maximum multi-commodity flow $f_{max}$ with value $\nu_{max}$ refers to the total flow over all possible paths that each commodity accepts. Assume instead that we focus on the maximum flow that can go through a specific network node $w\neq s,t$. Let $\mathcal{P}_i^w$ be the set of all $s_i-w-t_i$ paths (i.e. $s_i-t_i$ paths that go through $w$), and $\mathcal{P}_{i,e}^w$ the set of all $s_i-w-t_i$ paths that also go through edge $e$. The path-based formulation then is:
\begin{align*}
\text{maximize}\qquad & \nu^w=\sum_{i=1}^L\sum_{p\in \mathcal{P}_i^w} f_i(p)\\
\text{subject to }\qquad & \sum_{i=1}^L\sum_{p\in \mathcal{P}_{i,e}^w} f_i(p)\leq c(e), \forall e\in E\\
& \sum_{p\in \mathcal{P}_i^w}f_i(p)\leq D_i\\
& f_i(p)\geq 0, \forall i\in\{1,\dots,L\}, \forall p\in \mathcal{P}_i^w
\end{align*}

We denote the maximum flow through any node $w$ as the \textit{maximum} $w$-\textit{flow} $f_{max}^w$ and denote its value by $\nu_{max}^w$. Alternatively, we use the notation $s-w-t$ flow for single-commodity networks (or $\bm{s}-w-\bm{t}$ for multi-commodity networks). Similarly, for single-commodity flows we also write $\nu_{max}^w(s,t)$ (or $\nu_{max}^w(\bm{s},\bm{t})$ for multi-commodity networks) for the value of the maximum $w$-flow.

We emphasize three points. First, in the single-commodity case we always assume that $w\neq s,t$, even if not explicitly stated. Indeed, if either $w=s$ or $w=t$ then $\nu_{max}=\nu_{max}^w$. In this case, the problem is strongly polynomial and accepts combinatorial algorithms such as the Ford-Fulkerson algorithm. Second, to define the maximum $w$-flow we use paths, not simple paths. For traditional flow networks, this makes no difference as the maximum flow can be equivalently defined in terms of simple paths, paths or even walks, since we can always remove any cycles in the paths or walks that transmit flow to make them simple, without affecting the maximum flow. However, this is not the case with the maximum $s-w-t$ flow. As an example, consider the directed flow network in Figure \ref{figadd} where all edges have unit capacity. The maximum $s-w-t$ flow uses the directed path $s\to w\to s\to t$ for a value of 1, and no simple path exists for a $w$-flow. On the other hand, the traditional maximum $s-t$ flow can use the trivial simple path $s-t$, with a value that also happens to be 1 in this example; the cycle $s\to w\to s$ is redundant. Hence, even though the maximum $s-t$ flow can be equivalently defined either in terms of paths or simple paths, the $w$-maximum flow is different under the two definitions. In this work, we choose to use paths rather than simple paths. This point is discussed in more detail at the end of \cref{sec:directed}. Third, for undirected networks we allow the flow to be sent along any directed path from the source to the destination, as long as no directed edge is repeated. For example, consider the undirected network $w-s-t$ with capacities 1 for the two edges $(w,s)$ and $(s,t)$. In that case, the maximum $s-w-t$ flow is 0.5, and uses the directed path $s\to w\to s\to t$. Thus, the path carrying the flow can contain the same undirected edge twice, but it has to pass this edge in different directions when going from the source to the destination. An alternative definition that only permits end-to-end paths that can go though any edge at most once in any direction is briefly discussed at the end of \cref{sec:undirected}.
\begin{figure}[htbp]
	\begin{centering}
		\textsf{\includegraphics[width=0.55\linewidth]{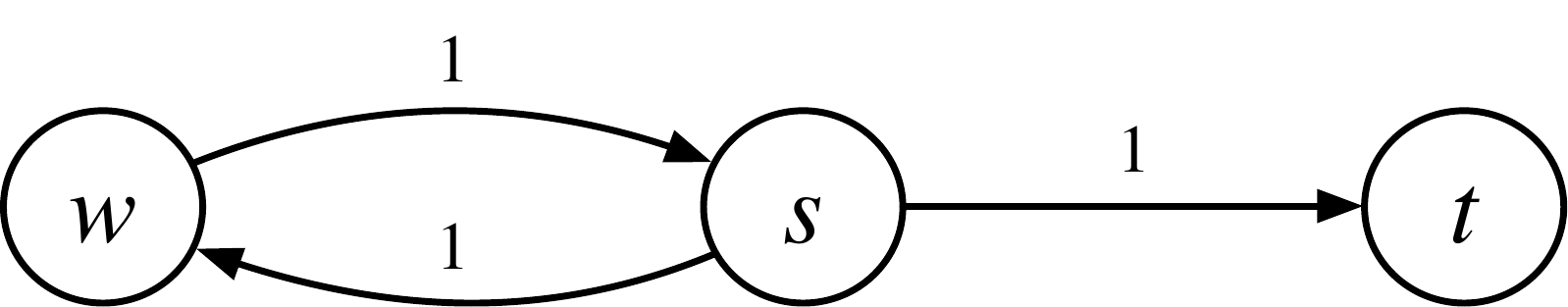}}
		\par\end{centering}
	
	\caption{The maximum $w$-flow with simple paths can differ from the maximum $w$-flow with paths.}
	\label{figadd}
\end{figure}

A central result in graph theory that we will be using throughout the paper is the two node-disjoint path (2DP) problem due to Fortune, Hopcroft and Wyllie \cite{fhw1980}.
\begin{theorem}[NP-hardness of 2DP \cite{fhw1980}]
\label{theorem:2node-disjoint}
Assume a directed graph $G=(V,E)$ and four distinct vertices $u_1,u_2,$ $v_1,v_2\in V$. It is NP-hard to decide whether there are two node-disjoint paths in $G$ from $u_1$ to $u_2$ and from $v_1$ to $v_2$.
\end{theorem}


We now provide two lemmas. Some of the transformations involved are standard in the disjoint-path literature (e.g., \cite{Lapaugh1980,schrijver-book}). Nevertheless, given the central role of the two lemmas in the remainder of this paper and in order to make our work self-contained, we provide our own full proofs here.

\begin{lemma}
\label{lemma:exists-simple-path-NPhard}
Deciding whether there exists a simple $s-w-t$ path in a directed graph $G=(V,E)$, where $w,s,t$ are three distinct nodes in $V$, is NP-hard.
\end{lemma}
\begin{proof}
Finding whether there is a simple $s-t$ path going through a node $w$ is equivalent to determining whether there exist two node-disjoint paths from $s$ to $w$ and from $w$ to $t$ (excluding of course node $w$). We prove that the latter problem is NP-hard by a reduction from the NP-hard 2DP problem.

Consider a directed graph $G=(V,E)$ and 4 distinct nodes $u_1,u_2,v_1,v_2\in V$. We introduce a new node $w$ and create the directed edges $e_1=(u_2,w)$ and $e_2=(w,v_1)$. We now argue that there are two node-disjoint paths, path $P_1$ from $u_1$ to $u_2$ and path $P_2$ from $v_1$ to $v_2$, if and only if there is a simple $u_1-w-v_2$ path. 
First, assume the former condition is true. Then $P_1$ cannot go through node $w$ via edge $e_1$; otherwise, that path would also have to use node $v_1$ after $w$ given $u_2$ is the end node. Similarly, we argue that $P_2$ cannot go through node $w$ via edge $e_2$. But then we can form a new path $P'$ from $u_1$ to $v_2$ by concatenating path $P_1$, edge $e_1$, edge $e_2$, and path $P_2$. $P'$ does not repeat any node since the node disjoint paths $P_1$ and $P_2$ do not contain $w$, hence it is a simple path. For the reverse direction, we note that if there exists a simple $u_1-w-v_2$ path $P$, then $P$ will necessarily contain edges $e_1$ and $e_2$. By removing these two edges, we get two node-disjoint paths, one from $u_1$ to $u_2$ and another from $v_1$ to $v_2$, given that $P$ is simple.
\end{proof}

\begin{lemma}
\label{lemma:exists-path-NPhard}
Deciding whether there exists a $s-w-t$ path in a directed graph $G=(V,E)$, where $w,s,t$ are three distinct nodes in $V$, is NP-hard.
\end{lemma}
\begin{proof}
Deciding whether there is a $s-t$ path going through  $w$ is equivalent to deciding whether there are two edge-disjoint paths from $s$ to $w$ and from $w$ to $t$. We argue that the latter problem is NP-hard by a reduction from the 2DP problem.

Indeed, consider a graph $G=(V,E)$, and three distinct nodes $s,t,w\in V$. We construct a new graph $G'=(V',E')$ from $G=(V,E)$ as follows. For each node $v\in V$ we introduce two nodes $v_{in},v_{out}\in V'$ as well as an edge $e'=(v_{in},v_{out})\in E'$ connecting them. For each edge $e=(u,v)\in E$, we introduce an edge $e'=(u_{out},v_{in})\in E'$. 
 The construction is illustrated in Figure \ref{fig:1}. 

\begin{figure}[htbp]
	\begin{centering}
	\subfloat[$G$]{\includegraphics[width=1.05\linewidth]{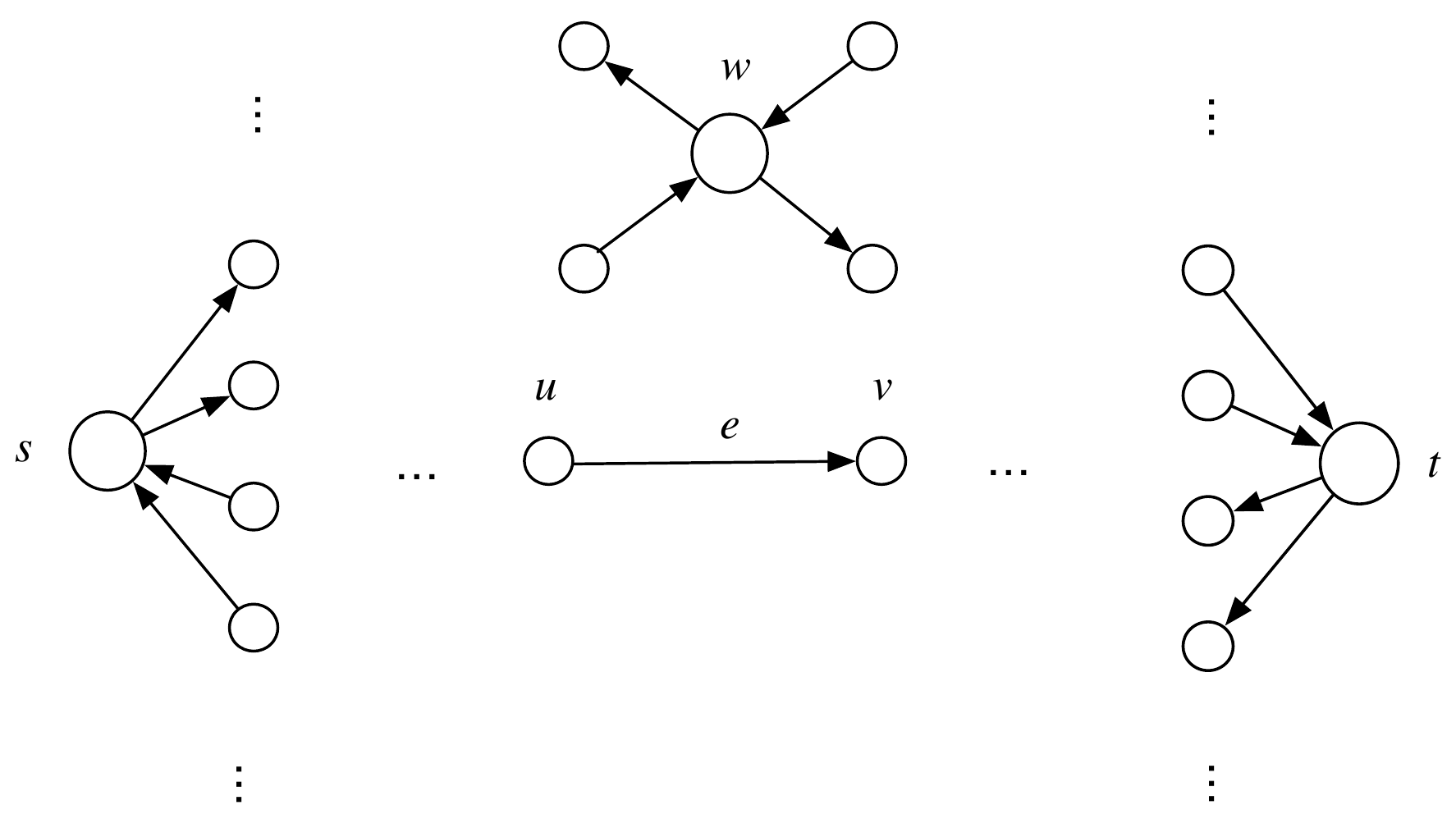}}\\
	\subfloat[$G'$]{\includegraphics[width=1.1\linewidth]{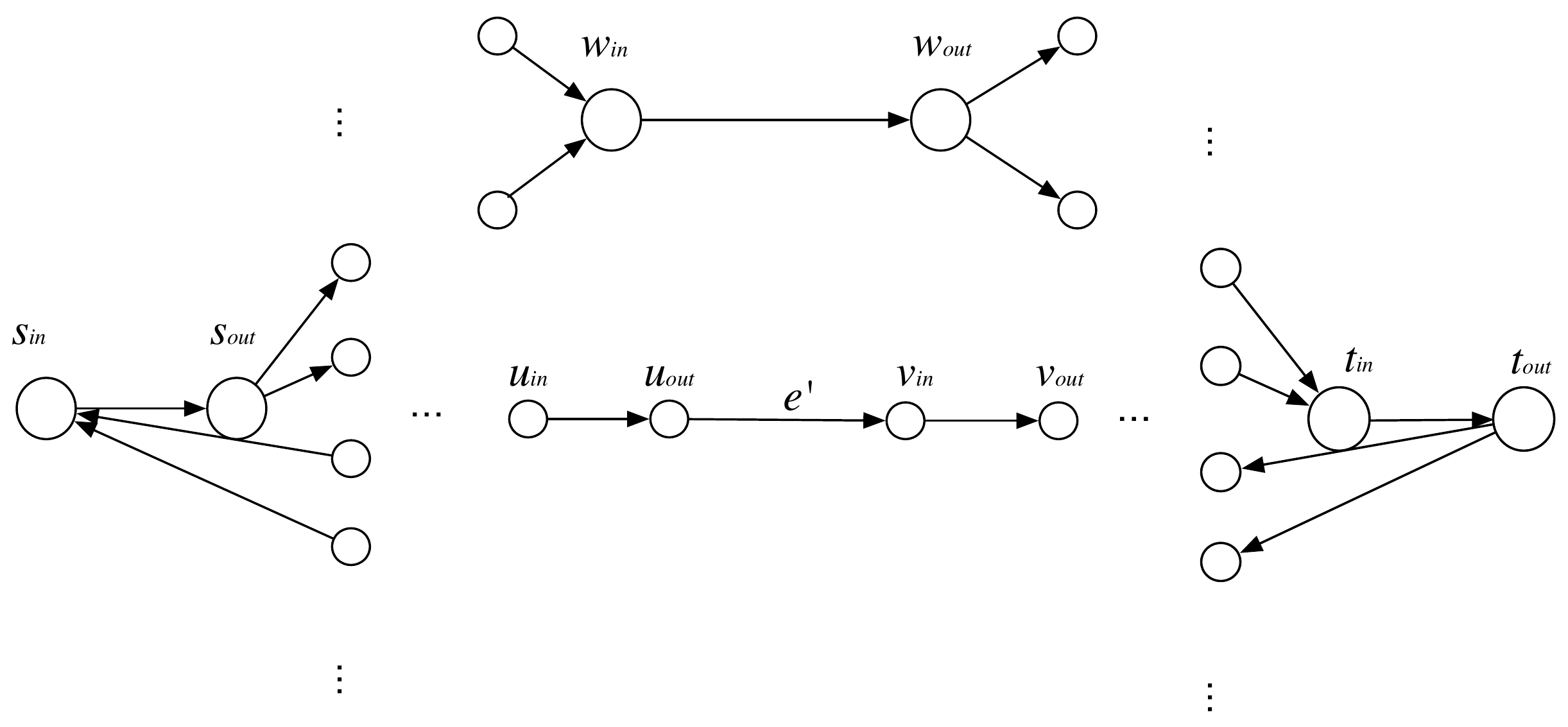}}
	
	\caption{Illustration of Lemma~\ref{lemma:exists-path-NPhard}.}	
	\label{fig:1}
	\end{centering}
\end{figure}

We now claim that there exist two edge-disjoint paths in graph $G'$ from $s_{out}$ to $w_{in}$ and from $w_{out}$ to $t_{in}$ (equivalently, from $s_{in}$ to $w_{in}$ and from $w_{out}$ to $t_{out}$), if and only if there exist two node-disjoint paths in $G$ from $s$ to $w$ and from $w$ to $t$.
First, consider two node-disjoint paths in $G$, namely, $s, u_1,\dots,u_l,w$ and $w,v_1,\dots,v_m,t$, where all intermediate nodes $u_i$ and $v_j$ are distinct. It is easy to see that the paths $s_{out},$ $u_{1,in},$ $u_{1,out},$ $\dots,$ $u_{l,in},$ $u_{l,out},$ $w_{in}$ and $w_{out},$ $v_{1,in},$ $v_{1,out},$ $\dots,$ $v_{m,in},$ $v_{m,out},$ $t_{in}$ in $G'$ are : (1) valid since they use existing edges in $G'$, and (2) edge-disjoint since the set of nodes on the first path and the second path are disjoint.

For the reverse direction, consider two edge-disjoint paths in $G'$ from $s_{out}$ to $w_{in}$ and from $w_{out}$ to $t_{in}$. We then argue that these paths must have the previous form $s_{out},$ $u_{1,in},$ $u_{1,out},$ $\dots,$ $u_{l,in},$ $u_{l,out},$ $w_{in}$ and $w_{out},$ $v_{1,in},$ $v_{1,out},$ $\dots,$ $v_{m,in},$ $v_{m,out},$ $t_{in}$. The reason is that any pair of nodes $(v_{in},v_{out})$ can only be reached from other nodes in $V'$ via $v_{in}$ and can only reach other nodes in $V'$ via $v_{out}$. So, a path will necessarily consist of consecutive pairs of nodes of the form $(v_{in},v_{out})$ (with the exception of the two endpoints). Furthermore, any such pair $(v_{in},v_{out})$ can (1) appear at most once on either path, and (2) cannot appear on both paths. The reason is that going from $v_{in}$ to $v_{out}$ requires edge $v_{in},v_{out}$, but the two paths are edge-disjoint. We thus conclude that  $s, u_1,\dots,u_l,w$ and $w,v_1,\dots,v_m,t$ in $G$ are node-disjoint paths. 
\end{proof}

We next provide definitions and results for the maximum $w$-flow that are reminiscent of results in traditional single-commodity maximum flow. For this purpose, we adapt some standard concepts from maximum flow theory \cite{amo1993}. In particular, we extend the concept of an $s-t$ \textit{augmenting path} to an $s-w-t$ \textit{augmenting path}, which corresponds to a directed path from $s$ to $t$ through middlepoint $w$ in the residual network. One significant difference is that the \textit{cut} is now defined as a collection of edges rather than a collection of nodes. The reason for this will become apparent shortly. We focus first on the $s-w-t$ flow in single-commodity networks.

\begin{definition}
A $s-w-t$ edge-cut is a subset of edges $\mathcal{C}^w\subseteq E$ such that removing the edges in $\mathcal{C}^w$ from the graph results in no $s-w-t$ paths, i.e., there are no $s-w-t$ paths in the graph $G'=(V,E-\mathcal{C}^w)$. The value $c(\mathcal{C}^w)$ of the edge-cut is defined as the sum of the capacities of all edges in $\mathcal{C}^w$.
\end{definition}

\begin{lemma}
\label{anyf_anyc}
Let $f^w$ be any $s-w-t$ flow, and $\mathcal{C}^w$ any $s-w-t$ cut. Then $\nu^w(f^w)\leq c(\mathcal{C}^w)$.
\end{lemma}
\begin{proof}
First, note that the flow $f^w$ is the sum of individual subflows, each going through a distinct $s-w-t$ path $p$. Each of these individual subflows must go through at least one edge $e\in\mathcal{C}^w$, otherwise there would be a $s-w-t$ path in the graph $G'=(V,E-\mathcal{C}^w)$, which would be a contradiction. So, let $\mathcal{F}_e$ be the set of subflows that go through $e$. Then we have:
\begin{align}
\sum_{e\in\mathcal{C}^w}\nu^w(\mathcal{F}_e) &\leq \sum_{e\in\mathcal{C}^w}c(e) \Leftrightarrow \notag \\
\sum_{e\in\mathcal{C}^w}\nu^w(\mathcal{F}_e) &\leq c(\mathcal{C}^w) \Leftrightarrow \notag \\
\nu^w(f^w) &\leq c(\mathcal{C}^w)
\end{align}
Note that in the last inequality we use that $\sum_{e\in\mathcal{C}^w}\nu^w(\mathcal{F}_e)$ $=$ $\nu^w(f)$, due to the fact that the path for each individual subflow must go through at least one edge in $\mathcal{C}^w$.
\end{proof}

\begin{lemma}
\label{maxc_minc}
Given a directed graph $G=(V,E,c)$ with integral capacities and three distinct nodes $s,w,t$, we can construct an integral $w$-flow. Furthermore, the constructed flow is positive if and only if the minimum edge-cut in $G$ is non-empty.
\end{lemma}
\begin{proof}
Consider the variant of the well-known Ford-Fulkerson algorithm for (single-commodity) maximum flow \cite{ps1982}, where at each round the algorithm picks an augmenting $s-w-t$ path rather than a $s-t$ path in the residual graph, provided that the selected augmenting path increases the flow through $w$. 
This is necessary, as it is possible to pick an augmenting $s-w-t$ path that uses reverse edges through $w$ and thus reduces the $w$-flow or leaves it unchanged, even though it increases the total $s-t$ flow.
The augmenting $s-w-t$ path algorithm eventually terminates, since (i) the $w$-flow increases by at least one unit at each iteration, and (ii) the maximum possible $w$-flow is upper-bounded (e.g., by the sum of capacities of the outgoing edges from $s$). Note that the augmenting $s-w-t$ path algorithm terminates, if and only if there is an $s-w-t$ edge-cut $\mathcal{C}^w$ in the graph where each edge $e\in\mathcal{C}^w$ is saturated. Indeed, if that were not true then there would be a $s-w-t$ path not using saturated edges. But then we could route more $w$-flow along the forward edges of that path in the residual graph, and the algorithm would not have terminated.
The constructed flow is integral, since at each step the flow on any edge is integral. Furthermore, the flow will be zero, if and only if the minimum edge-cut is the empty set, since in that case there are no $s-w-t$ paths in the original graph $G$.
\end{proof}

\begin{remark}
\label{remark-1}
Note that the variant of Ford-Fulkerson of Lemma \ref{maxc_minc} may not find the maximum $w$-flow. For example, consider the directed graph of Figure \ref{example-remark}, where all edges have infinite capacity except for edges $s\to w$, $w\to t$, and $u\to v$ with capacity 2. The algorithm could pick the augmenting path $s\to w\to t$ with bottleneck capacity 2. In the second iteration, there is no $s-w-t$ augmenting path in the residual graph, and the algorithm terminates returning a $w$-flow of value 2. Interestingly, the maximum $w$-flow in this example is equal to 3: we send one unit of flow along path $s\to w\to t$, one unit of flow along path $s\to w\to v\to t$, and finally one unit of flow along $s\to u\to v\to w\to t$.
\end{remark}

\begin{remark}
\label{remark-2}
For node-constrained flow, it is not true that the value of the minimum $s-w-t$ edge-cut is equal to the value of the maximum $w$-flow. For example, in the graph of Figure \ref{example-remark}, the minimum $s-w-t$ edge-cut has capacity 4 (e.g., edges $s\to w$ and $w\to t$), but the maximum $w$-flow has a value of 3. Nevertheless, Lemma \ref{anyf_anyc} guarantees the minimum $s-w-t$  cut upper-bounds the value of the maximum $w$-flow.
\end{remark}

\begin{remark}
\label{remark-3}
Even if all capacities are integral, the maximum $w$-flow may be fractional. For example, consider the graph of Figure \ref{example-remark} but assume the edges with capacity 2 now have unit capacity. The maximum $w$-flow has value $\frac{3}{2}$, which is fractional.
\end{remark}

The properties in Remarks \ref{remark-1}, \ref{remark-2}, and \ref{remark-3} imply that $s-w-t$ flow is fundamentally different from the traditional $s-t$ flow.
\begin{figure}[t]
	\begin{centering}
		\textsf{\includegraphics[width=0.55\linewidth]{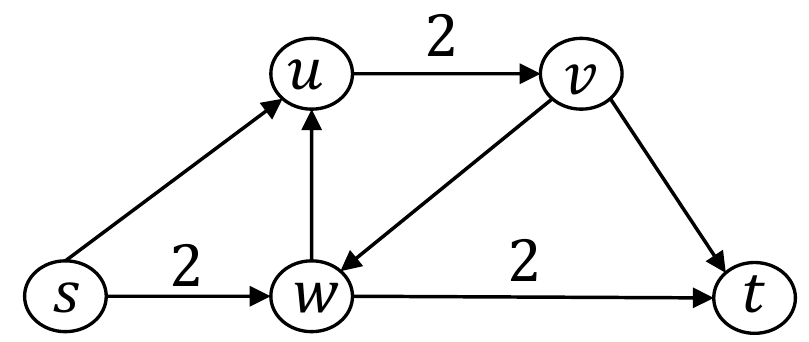}}
		\par\end{centering}
	
	\caption{Example Graph for Remarks \ref{remark-1}, \ref{remark-2}, and \ref{remark-3}.}
	\label{example-remark}
\end{figure}
We are now ready to prove that the decision version of the maximum $w$-flow is NP-hard.

\begin{proposition}
\label{prop:NPhardness}
Given a multi-commodity flow network $G=(V,E,c)$ with directed edges, the decision version of maximum $w$-flow is NP-hard.
\end{proposition}
\begin{proof}
We show that even the single-commodity version is NP-hard. Our strategy is to reduce the $s-w-t$ path problem in Lemma \ref{lemma:exists-path-NPhard} to the maximum $w$-flow problem.
In this direction, we start with a directed graph $G=(V,E)$ and three distinct nodes $s,t,w\in V$. We subsequently construct in polynomial time a flow network $G'$ from $G$ by considering a single commodity from $s$ to $t$ of unit demand $D=1$, and by associating each edge $e\in E$ with a unit capacity. Our claim is that there is a path from $s$ to $t$ through $w$ in graph $G$, if and only if $\nu_{max}^w(s,d)\geq 1$ in flow network $G'$.

First, consider a path $P=(s,e_1,\dots,e_m,t)$, so that every edge $e_i$ in the path appears only once and node $w$ appears in the path. It is then possible to send one unit of flow from $s$ to $t$, given the unit capacities. Thus, the maximum flow is at least 1. For the reverse direction, assume there is a maximum flow no less than 1 in $G'$. 
Since we only have one commodity and integral capacities, Lemma \ref{maxc_minc} implies that there exists an integral $w$-flow (for infinite demands). This flow is positive (and thus has value at least 1) because the maximum flow is no less than 1 by assumption, and thus the minimum edge-cut cannot be the empty set (recall we route flow along $s-w-t$ paths). 
For $D=1$, there is thus an integral maximum $w$-flow of value 1.
Moreover, each edge can be used at most once in that flow due to its unit capacity. Now, consider any path $P_{(s,w)}$ that carries the integral flow from $s$ to $w$, and any path $P_{(w,d)}$ that carries the integral flow flow from $w$ to $d$. Since each edge is used at most once, this means that the union $P_{(s,w)}\cup P_{(w,d)}$ (1) goes from $s$ to $d$ through $w$, and (2) visits any edge at most once; hence, it is a $s-w-t$ path.
\end{proof}

\subsubsection{Hardness of $TE_{LU}$}
\label{sec:hardness_TE}

For the flow network $G$ with the $L$ commodities, we define the $TE_{LU}$ in the following manner:
\begin{align}
\min\qquad & \theta\label{TE2-LU-obj}\\
\text{subject to }\qquad & \sum_{i=1}^L\sum_{p\in \mathcal{P}_{i,e}^w} f_i(p)\leq \theta\cdot c(e), \forall e\in E\label{TE2-LU-first-constraint}\\
& \sum_{p\in \mathcal{P}_i^w}f_i(p)\geq D_i, \forall i\in\{1,\dots,L\}\label{TE2-LU-second-constraint}\\
& f_i(p)\geq 0, \forall i\in\{1,\dots,L\}, \forall p\in \mathcal{P}_i^w\label{TE2-LU-last-constraint}
\end{align}
We get the following corollary:
\begin{corollary}
\label{corollary:NP-hard-TE-LU}
In directed graphs, it is NP-hard to solve the $TE_{LU}$ \eqref{TE2-LU-obj}--\eqref{TE2-LU-last-constraint}.
\end{corollary}
\begin{proof}
Proposition \ref{prop:NPhardness} shows NP-hardness for the decision version of the maximum $w$-flow.
We can then easily extend Lemma \ref{lemma:connection} to the case of a maximum $w$-flow to prove that the $TE_{LU}$ \eqref{TE2-LU-obj}--\eqref{TE2-LU-last-constraint} is NP-hard, since it is at least as hard as the decision version of the maximum $w$-flow.
\end{proof}

\textbf{Simple paths vs. paths vs. walks.}
Finally, we discuss in more detail our definition of $w$-flow. 
In principle, we can define the $w$-flow in terms of simple paths, paths or general paths with edge repetitions (walks).
For walks, the maximum $w$-flow is in fact polynomial.
To see why, consider the single-commodity $s-t$ maximum flow problem with middlepoint $w$ (the multi-commodity case is similar). We first solve the multi-commodity flow problem with source-destination pairs $(s,w)$ and $(w,t)$, maximizing the minimum of  the amount of two flows. We can subsequently get $s-w-t$ paths after decomposing the two flows into paths. 

Interestingly, our definition excludes walks.
The reason for this choice is twofold. 
First, in computer networks routing loops are considered bad and are avoided by most routing algorithms, since they can lead to redundant use of the precious bandwidth resources and, more seriously, to endless routing loops and failure of packet delivery \cite{KR16}. Thus, in practical networking settings TE with simple paths is the primary option. 
Second, our work is inspired to a significant degree by prior work on flow centrality \cite{fbw1991} which explicitly defines flow in terms of simple paths. 
Based on Lemma \ref{lemma:exists-simple-path-NPhard}, we can show similar to the case of paths that the decision problem of $w$-flow with simple paths is also NP-hard in directed graphs. 
This suggests there is no difference in terms of hardness if we use simple paths or paths; we select the latter because paths are more general than simple paths. 



\subsection{The Undirected Case}
\label{sec:undirected}

We demonstrate an interesting dichotomy between directed and undirected graphs. In detail, the maximum multi-commodity $\bm{s}-w-\bm{t}$ flow in an undirected graph can be computed in strongly polynomial time. Note that the main difference between the directed and the undirected flow is that the former assumes separate capacities for each direction $(u,v)$ and $(v,u)$ whereas the latter assumes a single capacity for the undirected edge $e$, which upper bounds the \textit{total} flow that we can send in both directions (but not in any individual direction).

\begin{proposition}
\label{prop:undirected-poly}
The maximum multi-commodity flow $\nu_{max}^w$ in any flow network $G=(V,E)$ with undirected edges, where $w\in V$, can be computed exactly in strongly polynomial time.
\end{proposition}
\begin{proof}
Assume a multi-commodity undirected graph $G=(V,E)$ with $L$ commodities of the form $(s_i,t_i)$. For simplicity, we assume infinite maximum demands $D_i$ (so that the maximum demand constraints are trivially satisfied). 
To prove the claim, we construct a directed graph $G'=(V',E')$ from $G$ as follows. We first replace each undirected edge $(u,v)\in E$ by two directed edges $(u,v)$ and $(v,u)$ edges with infinite capacities. We next introduce $L$ new nodes $z_1,\dots,z_L$ (one for each commodity), and for each $z_i$ we add the two directed edges $(s_i,z_i)$ and $(t_i,z_i)$. Finally, we introduce a node $z$ and $L$ directed edges $(z_i,z)$ from each $z_i$ to $z$. Thus, we have that $V'=V\cup\{z_1,\dots,z_L,z\}$ and $E'=E_G\cup(\cup_{i}\{(s_i,z_i),(t_i,z_i),(z_i,z)\})$, where $E_G$ are the edges that we got by replacing each undirected edge in $E$ by two directed edges.  The capacities of the newly constructed edges of the form $(s_i,z_i),(t_i,z_i),(z_i,z)$ are infinite. The construction is illustrated in Figure \ref{fig:2}.

Next, we claim that the maximum flow in $G$ can be computed by  the following arc-based linear program for $G'$:
\begin{align}
\text{maximize } & \mathcal{V}=\sum_{i=1}^L\sum_{e\in {w}^{+}}f_i(e)\notag\\
\text{subject to } & \sum_{i=1}^L [f_i(u,v) + f_i(v,u)]\leq c(e), \forall e=(u,v)\in E\label{c1}\\
& \sum_{e\in u^+}f_i(e)=\sum_{e\in u^-}f_i(e), \forall i, \forall u\in V', w\neq u\neq z\label{c2}\\
& f_i(e)\geq 0, \forall e\in E', \forall i\in\{1,\dots,L\}\label{c3}\\
& f_j(s_i,z_i) = 0, \forall i,j\in\{1,\dots,L\}\text{ with } i\neq j\label{before-last}\\
& f_i(s_i,z_i) = f_i(t_i,z_i), \forall i\in\{1,\dots,L\}\label{last}
\end{align}

\begin{figure}[th]
	\begin{centering}
	\subfloat[$G$]{\includegraphics[width=0.77\linewidth]{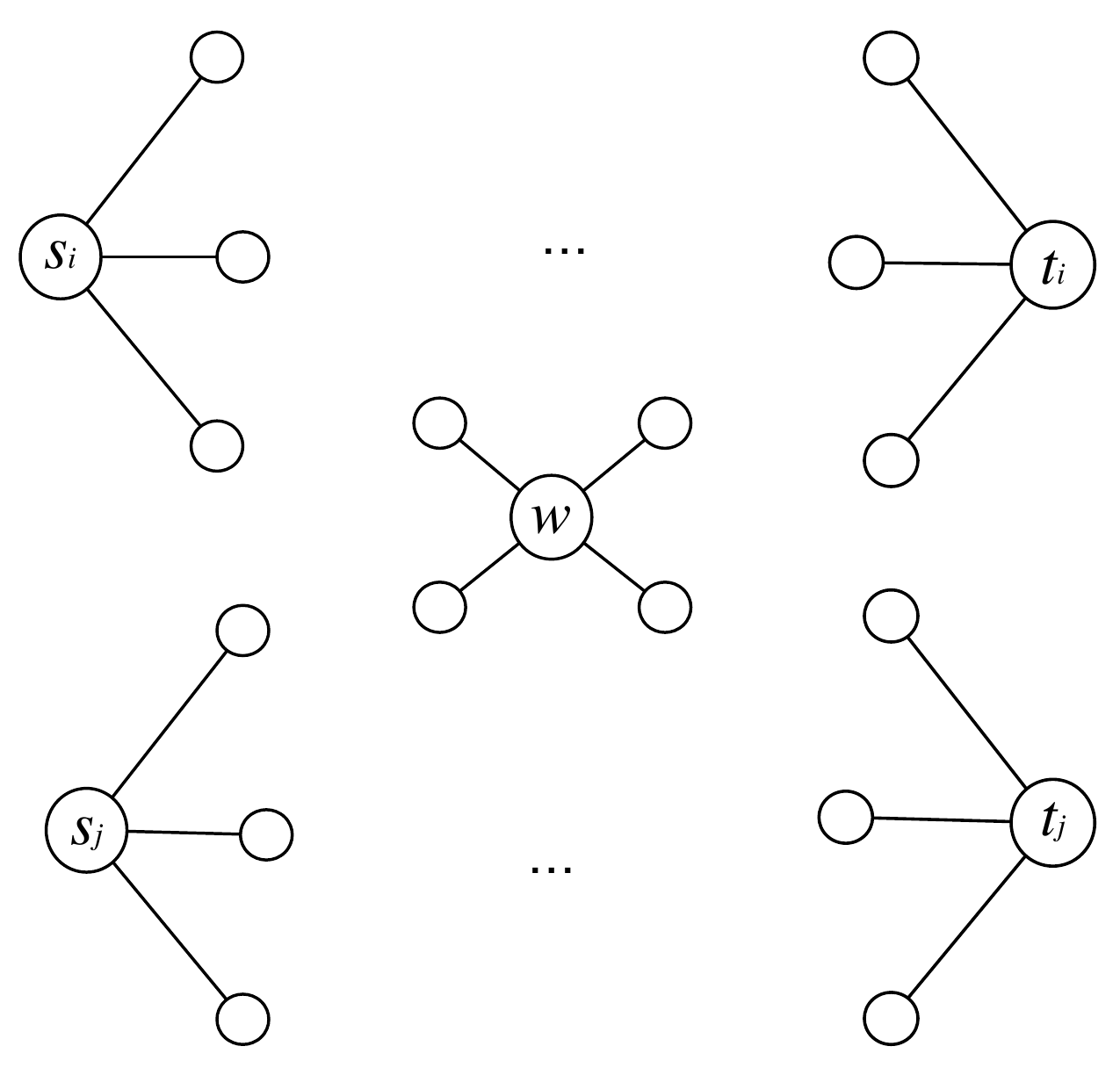}}
	\vspace{-3mm}
	\hspace{0.01in}
	\subfloat[$G'$]{\includegraphics[width=0.77\linewidth]{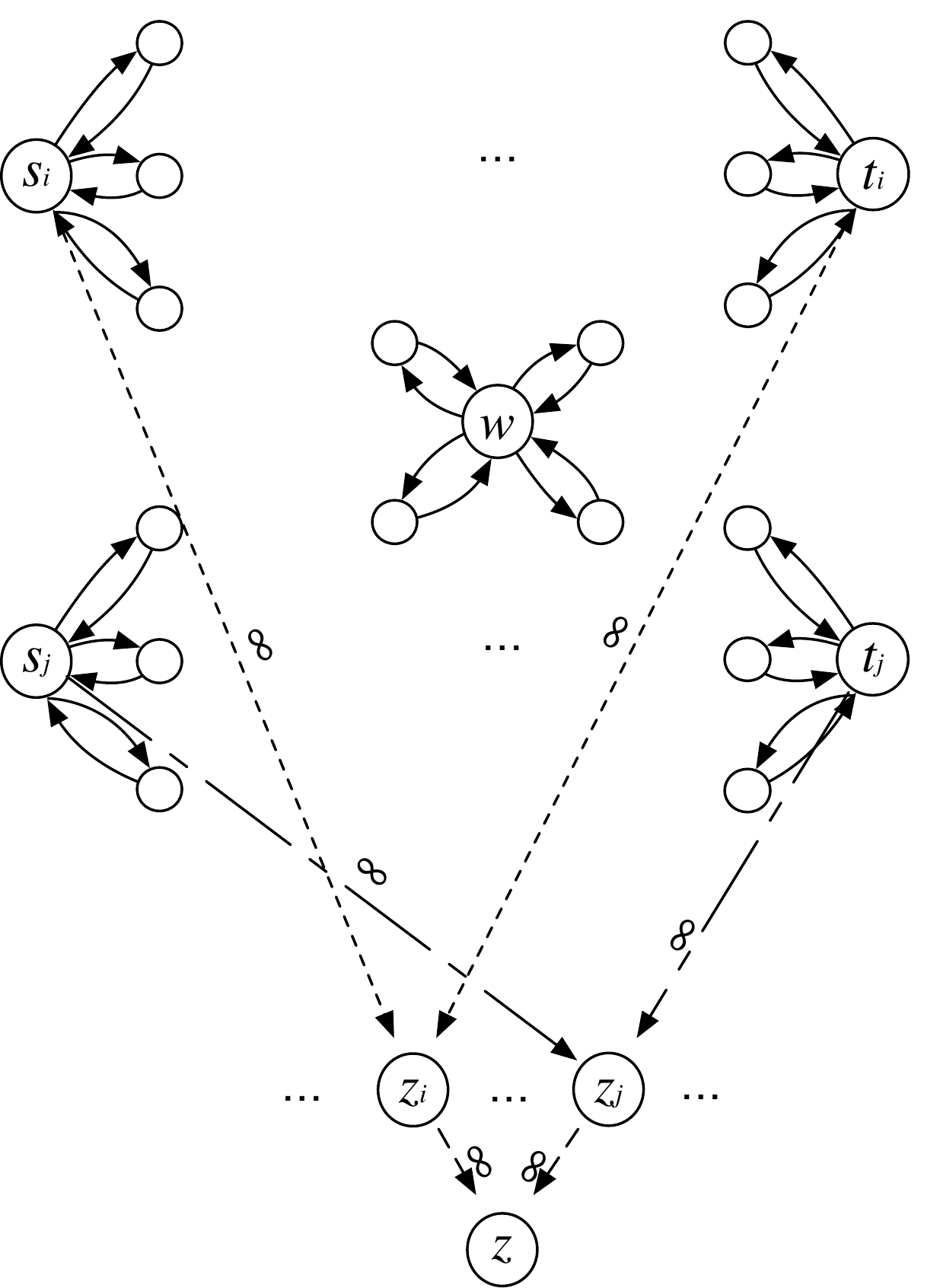}}
	\caption{Illustration of Proposition~\ref{prop:undirected-poly}.}	
	\label{fig:2}
	\end{centering}
\end{figure}

The above LP computes the maximum flow $\mathcal{V}^*$ from $w$ to $z$, where the total flow is composed of $L$ separate subflows, one for each commodity. The subflow for commodity $i$ can be sent from $w$ to $z_i$ through either $s_i$ or $t_i$. Constraints \eqref{c1}, \eqref{c2}, \eqref{c3} are the link capacity, node conservation and positive flow constraints, respectively. The link capacity constraint reflects the fact that the sum of flow units in each of the two directed edges does not exceed the capacity of the original undirected edge. Constraint \eqref{before-last} implies that node $z_i$ can only receive flow from commodity $i$. Constraint \eqref{last} is especially important because it ensures that the subflow for commodity $i$ sent through $s_i$ is the same as the one sent through $t_i$.

Now, we establish the equivalence between the original problem and the above LP by showing that (i) $2\cdot \nu_{max}^w\leq \mathcal{V}^*$, and (ii) $\mathcal{V}^*\leq 2\cdot \nu_{max}^w$. We start with (i). Assume any maximum flow through $w$ with value $\nu_{max}^w$ in $G$. We first construct the corresponding flow in $G'$. We note that the flow in $G$ consists of $L$ subflows, one for each commodity $i$, that send flow from $s_i$ to $t_i$ through $s_i-w-t_i$ paths. The idea is then to reverse the direction of each subflow in the part from $s_i$ to $w_i$, so that it now sends to the opposite direction. We then send $\nu(f_i)$ units of flow from $s_i$ to $z_i$, $\nu(f_i)$ units of flow from $t_i$ to $z_i$, and $2\cdot \nu(f_i)$ units of flow from $z_i$ to $z$. Note that that is a valid flow since it respects all constraints in the LP, and it has a value $2\cdot(\nu(f_1)+\cdots+\nu(f_L))=2\cdot\nu_{max}^w$. But then the maximum flow will be at least as large, hence $2\cdot \nu_{max}^w\leq \mathcal{V}^*$.

For the reverse direction (ii), assume a maximum flow in $G'$. Then for each commodity $i$ half units are sent from $w$ to $s_i$ and half from $w$ to $t_i$ (and subsequently $z_i$ and $z$) due to constraint \eqref{last}. 
As before, we next reverse the flow $f_i$ in all paths from $w$ to $s_i$. For each edge of $G$ we then send on each direction as many units of flow as we send in $G'$ (after reversing the direction from $w$ to $s_i$). The key is that in graph $G'$ the same amount of flow is sent from $w$ to $s_i$ and $w$ to $t_i$, which implies that the constructed flow in $G$ will respect all capacity and conservation constraints, including for node $w$. 
Note however a caveat: it is not obvious whether we can always combine directed paths from $s_i$ to $w$ with directed paths from $w$ to $t_i$ 
so that \textit{all} resulting $s_i-w-t_i$ ``paths'' are valid paths without edge repetitions. If that is possible, then the constructed flow in $G$ is a valid flow, since it consists of $s_i-w-t_i$ paths without edge repetitions. But we need to address the case where combining directed paths from $s_i$ to $w$ with directed paths from $w$ to $t_i$ will result in at least one non-valid $s_i-w-t_i$ path with edge repetitions. 

To show why this does not negatively affect our argument, let $P_1$ be \textit{any} $s_i-w$ and $P_2$ \textit{any} $w-t_i$ path in $G'$ whose concatenation is not a valid $s_i-w-t_i$ path. In this case, let $e=(u,v)$ be a common edge of $P_1$ and $P_2$ (such an edge must occur by assumption). Path $P_1$ can then be represented as $s_i\leadsto_{P_1} u \to v\leadsto_{P_1} w$, where the notation $x\leadsto_{P} y$ refers to the path segment from node $x$ to node $y$ along path $P$. Note that we may have that $x$ coincides with $y$, in which case $x\leadsto_{P} y= x$. Similarly, path $P_2$ can be represented as $w\leadsto_{P_2} u \to  v\leadsto_{P_2} t_i$. The trick is to consider the $s_i-w-t_i$ path $P'=s_i\leadsto_{P_1}u\leadsto_{P_2^{-1}}w\leadsto_{P_1^{-1}}v\leadsto_{P_2}t_i$
, where $x\leadsto_{P^{-1}} y$ refers to the path from $x$ to $y$ along the reversed edges of $P$.
We first discuss the case where $P'$ is a valid path without edge repetitions. We then send in graph $G$ as many units of flow along $P'$ as we did along $P_1$ or $P_2$ in $G'$. The critical observation is that an undirected edge in $G$ will receive \textit{at most} as much flow as the sum of its two directed edges in $G'$, since all we do is change the direction along which we send the flow and possibly remove certain path segments. The trick is illustrated in Figure \ref{P4}. 
In the case where the resulting path $P'$ is not valid and contains edge repetitions, we repeat the step in Figure \ref{P4} on the new $s_i-w$ and $w-t_i$ paths until they have no common edge. The process is guaranteed to stop after a finite number of iterations since after each step the total length of the new $s_i-w$ and $w-t_i$ paths decreases.
\begin{figure}[t]
	\begin{centering}
		\subfloat[$G'$]{\includegraphics[width=0.62\linewidth]{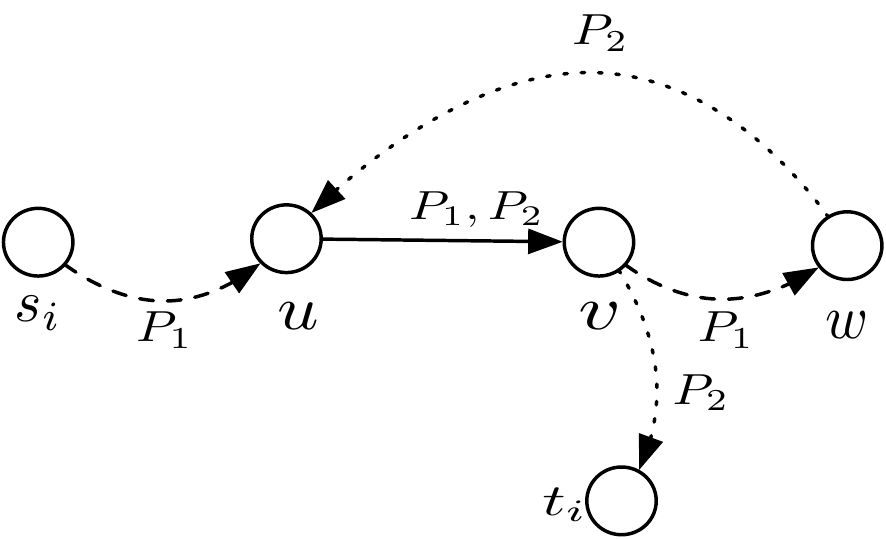}}
		\\
		\subfloat[$G$]{\includegraphics[width=0.62\linewidth]{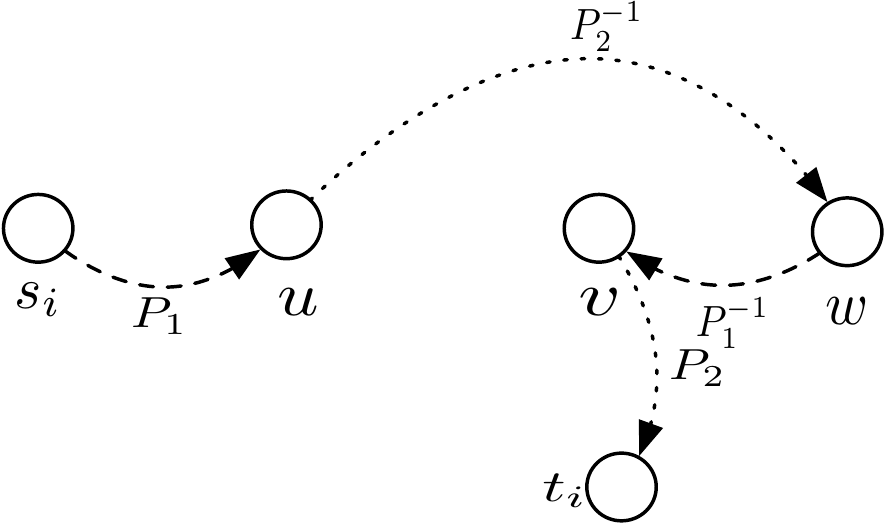}}
		\caption{Getting a valid $s_i-w-t_i$ path for undirected graph $G$ from invalid path in directed graph $G'$ in Proposition \ref{prop:undirected-poly}.}	
		\label{P4}
	\end{centering}
\end{figure}
In principle, we can do this for all non-valid $s_i-w-t_i$ paths to get $s_i-w-t_i$ paths which are valid, without violating any capacity constraint. 
Our argument shows that there must exist a valid multi-commodity $w$-flow in $G$, which carries the same amount of flow from $s_i$ to $t_i$ as the flow from $w$ to $s_i$ in $G'$ (note we do not actually need to explicitly construct that flow in $G$ as we just care about $\nu_{max}^w$).
Concretely, commodity $i$ sends in $G$ the same number of flow units as $w$ sends to $s_i$ (or $t_i$) in $G'$. 
Obviously, the value of the constructed $w$-flow in $G$ is half that in $G'$, thus for the maximum $w$-flow in $G$ we must have that $\mathcal{V}^*\leq 2\cdot \nu_{max}^w$.

Finally, given that (i) the constraints of the LP have a size that is polynomial in $|V|$ and $|E|$, and (ii) the constraint matrix of the LP only contains entries $-1,0,1$, we deduce that computing the maximum $\bm{s}-w-\bm{t}$ flow in an undirected graph can be solved in strongly polynomial time \cite{Tardos1986}.
\end{proof}

\textbf{An alternative type of undirected flow.}
We previously considered that the flow in undirected networks can pass an edge twice, albeit in different directions. What if we only permit end-to-end paths that can pass an edge at most once in any direction?

To this end, we first review a major result for the $k$-disjoint path ($k$DP) problem from the theory of graph minors by Robertson and Seymour \cite{Robertson1995}.
\begin{theorem}[$k$DP Problem]
\label{theorem:kdp}
Given a graph $G=(V,E)$ and $k$ pairs $(s_1,t_1),\dots,(s_k,t_k)$ of vertices of G, the $k$DP problem of deciding whether there exist pairwise vertex-disjoint or edge-disjoint paths $P_1,\dots,P_k$ such that $P_i$ connects $s_i$ and $t_i$ ($1\leq i\leq k$) is in P, when $k$ is fixed and not part of the input \cite{Robertson1995}. When $k$ is part of the input, both the vertex-disjoint and edge-disjoint decision problems are NP-complete \cite{Karp1975,Even1975}.
\end{theorem}

Note that even though the $k$DP problem is in P, the algorithm is not practically feasible, since it involves the manipulation of enormous constants \cite{Robertson1995}. Given that the computation of a $s-w-t$ path in undirected graphs is polynomially computable (despite the enormous constants involved), it is worthwhile to explore whether we can use this result in conjunction with the augmenting path algorithm.

\begin{remark}\label{remark:undirected}
We argue that the augmenting path algorithm cannot be directly used to compute the maximum $s-w-t$ flow when using end-to-end paths that cannot traverse the same edge twice (in different directions). Consider for example the undirected graph of Figure \ref{augmenting-undirected} where all undirected edges have infinite capacity except for edges $(v,w)$ and $(w,x)$ with capacity 2. We can initially construct a flow of 2 units using the end-to-end augmenting path $s\to v\to w\to x \to t$. But once we do that it is not possible to find other augmenting $s-w-t$ paths that can strictly increase the $w$-flow, so the algorithm returns a flow through $w$ of 2 units. Nevertheless, there is a higher $w$-flow of 3 units, which sends one unit of flow along each of the following end-to-end paths: $s\to v\to w\to x\to t$, $s\to u\to v\to w \to t$, and $s\to x\to w \to t$. 

It is an open question whether a polynomial algorithm exists for this second type of undirected $w$-flow or not.
\end{remark}
\begin{figure}[h]
	\begin{centering}	\textsf{\includegraphics[width=0.85\linewidth]{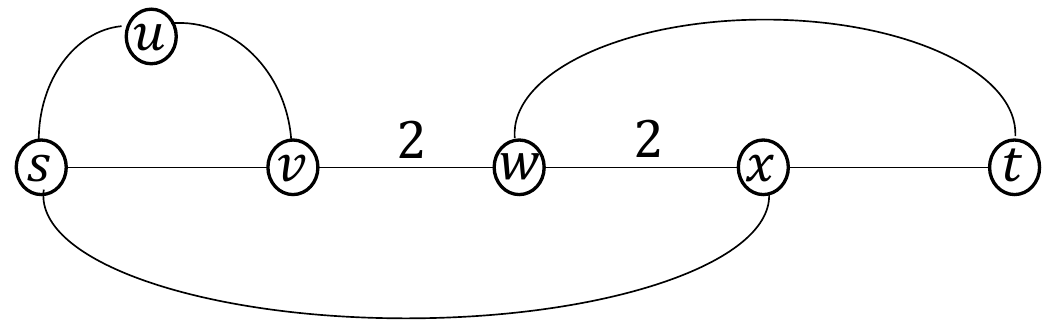}}
		\par\end{centering}
	\caption{Undirected Graph for Remark \ref{remark:undirected}.}
	\label{augmenting-undirected}
\end{figure}

\subsection{Extending to Many Middlepoints}

We conclude this section with a last remark. So far we have discussed node-constrained TE where the traffic has to go through a specific node $w$. 
This is mostly for analytical convenience.
What about the general case where the traffic can go through \textit{at least} one node from a set of more than one middlepoints? 
Obviously, for directed graphs the problem remains NP-hard. For undirected graphs, it is possible to show that when the number of nodes is fixed and not part of the input, then TE remains strongly polynomial. Indeed, we can generalize the proof of Proposition \ref{prop:undirected-poly} by assuming $k$ nodes $w_1,\dots,w_k$ and constructing the same directed graph $G'$. We can then show similar to Proposition \ref{prop:undirected-poly} that the maximum multi-commodity flow that goes through at least one node in the set $\{w_1,\dots,w_k\}$ is equal to half the maximum multi-commodity flow from the sources $w_i, 1\leq i\leq k,$ to their common destination $z$ in the constructed directed graph $G'$.
We thus rigorously establish that node-constrained TE is NP-hard and difficult to solve in general, as most TE problems use directed graphs to model bidirectional links and traffic.
\section{Node-Constrained Traffic Engineering \\with Shortest Paths}
\label{sec:opt} 

The previous section shows that general node-constrained TE is NP-hard in directed graphs. 
We now consider variants of node-constrained TE where only shortest paths between two middlepoints are used in TE, and investigate if these variants are  easier to solve.
In \cref{sec:model} we formulate node-constrained TE with shortest paths using a fixed number of middlepoints, and show that it is weakly polynomial.
In this sense, our results provide for the first time a theoretical foundation for existing work that focuses on shortest path based segment routing \cite{bhatia2015optimized,HVSB15}.\footnote{We emphasize that the problem of determining an \textit{optimal} set of middlepoints of a given size in segment routing is hard, and a large part of the prior literature on segment routing has tried to address that question. However, this problem is out of the scope of this work.}
Given that TE with shortest paths may result in end-to-end paths that contain cycles, in \cref{sec:SRhardness} we study a more specific variant of acyclic node-constrained TE with shortest paths, and show that it is generally NP-hard. Note that in this section we focus on directed graphs, since segment routing typically considers directed graphs to model network traffic.

\subsection{Variant with Shortest Paths}
\label{sec:model}



Assume $K$ middlepoints in total with a specific ordering, where $1\leq K\leq |V|$. We assume that each end-to-end path can use up to $M\leq K$ of these middlepoints respecting the ordering as the $K$ input middlepoints. Note that prior works typically assume small values of $M$; for instance $M$ can be as small as 1, in which case each end to end path consists of 2 segments \cite{bhatia2015optimized}. For a segment $s\in S$ between an ingress node and a middlepoint, two middlepoints, or a middlepoint and an egress node, there are multiple paths in general. We assume, for simplicity, that routing is done by ECMP over all shortest paths of a segment. ECMP routes a flow based on static hashing of the five tuples in the packet header, and in general can distribute traffic evenly when the number of flows is large. This is consistent with prior work \cite{bhatia2015optimized}. We use  $T_i$ to denote the complete set of logical tunnels formed by segments in $S$ that can be used for commodity $i$, with up to $M$ middlepoints. A tunnel involves only ingress/egress switch, and the intermediate middlepoints. This can be constructed offline efficiently. 

Let $G_{t,s}$ denote if a tunnel $t$ uses segment $s$ or not, 
and $I_{p,e}$ denote if path $p$ uses link $e$ or not. Furthermore, let $\hat{P}_s$ be the set of all shortest paths for segment $s$,
and $f_{i}(t)$ represent the flow in tunnel $t$ for commodity $i$. The split ratio  $x_{i,t}$ for $i$ on tunnel $t$ is defined as the ratio $x_{i,t}=\frac{f_i(t)}{\sum_{t \in T_i}f_{i}(t)}$.
The node-constrained $TE_{LU}$ problem with shortest paths can be formulated similar to \cref{sec:TE-LU}, where the set of paths $\mathcal{P}_i$ for commodity $i$ is now replaced by the set of logical tunnels $T_i$: 
\begin{align}
\min      & \quad\quad \theta \label{obj} \\  
\text{s.t. }& \sum_{i=1}^{L} \sum_{t \in T_i}\sum_{s\in S_t}\sum_{p\in \hat{P}_s} f_{i}(t)\frac{I_{p,e}}{\lvert \hat{P}_s \rvert}    \leq \theta\cdot c(e) ,\forall e \in E,  \label{con:capacity} \\
&  0 \leq f_{i}(t) ,\forall  i  \in \{1,\dots,L\},  t \in T_i, \label{con:nonnegative} \\
& \sum_{t \in T_i}f_{i}(t) \geq D_i, \forall  i  \in \{1,\dots,L\}. \label{con:flow}
\end{align}

The capacity constraint \eqref{con:capacity} indicates that the total traffic routed to link $e$ from across all flows, tunnels, segments, and shortest paths, cannot exceed $\theta$ times the link capacity. Since ECMP is used for routing within any segment $s$, each shortest path $p$ of segment $s$ receives flow equal to $\frac{f_{i}(t)}{\lvert \hat{P}_s \rvert}$. 
Regarding the TE asymptotic complexity, we have the following result when $M$ is fixed and not part of the input:

\begin{proposition}
\label{prop:te-sp-lu}
For fixed $M$ with respect to the input graph $G$, the $TE_{LU}$ problem described by \eqref{obj}-\eqref{con:flow} can be solved in weakly polynomial time.
\end{proposition}
\begin{proof}
The number of commodities $L$ cannot exceed $|V|\cdot(|V|-1)$, and the number $|T_i|$ of tunnels per commodity $i$ is upper bounded by 
${{K}\choose{0}}+\cdots+{{K}\choose{M}}$, where $K\leq|V|$. For fixed $M$ w.r.t. the input graph $G$, $|T_i|$ has polynomial size w.r.t the graph. Finally, the number $S_t$ of segments per tunnel cannot exceed $K+1\leq|V|+1$, since a tunnel can use at most all $K$ middlepoints. For the inner sum $\sum_{p\in \hat{P}_s} \frac{I_{p,e}}{\lvert \hat{P}_s \rvert}$, note that it basically denotes the percentage of shortest paths for segment $s$ that use link $e$. This percentage naturally appears in the definition of betweenness centrality \cite{B2001}, and we can compute it in polynomial time by using the techniques therein.\footnote{Even though the number of shortest paths between two nodes can be exponential, the percentage of shortest paths that go through a specific node is polynomially computable. This is why betweenness centrality is polynomially computable \cite{B2001}.}


Thus, we have proved that for fixed $M$, the LP has a polynomial number of variables, and a polynomial number of constraints whose coefficients can be computed in polynomial time. The proposition then immediately follows by standard results in linear programming \cite{K1984,K1980}.
\end{proof}

Given the connection between $TE_{MF}$ and $TE_{LU}$, we can similarly show that $TE_{MF}$ can be solved in weakly polynomial time.
As we observed, the TE problem is naturally related to the betweenness centrality that we discuss in \cref{rel:centrality}. 
Thus, constraint \eqref{con:capacity} reveals interesting connections between the popular centrality metric and node-constrained TE with shortest paths.

Note that the TE is not polynomial just for the specific setting where there are $K$ middlepoints and each logical path can use up to $M$ of them. In particular, it remains polynomial if we constrain each logical path to contain all $K$ middlepoints in the same order (instead of up to $M$), where $K$ can be as large as $|V|$. Interestingly, even this simple variant can become NP-hard if we impose simple constraints, as we show next.

\subsection{Acyclic Variant with Shortest Paths}
\label{sec:SRhardness}
One challenge with node-constrained TE with shortest paths is that it may produce source-destination paths with edge repetitions, i.e., walks. Even in the case of just one middlepoint per path, it is possible that a (simple) shortest path from the source $s$ to a middlepoint $M$ shares an edge $e$ with a shortest path from $M$ to the destination $d$. So, even though the paths for any given segment are simple, the resulting $s-d$ path may not even be a path. This may reduce the performance of TE, because it increases the link load on the reused edges and may lead to higher link utilization. Figure \ref{cycle} depicts a directed graph with $n>2$, where the shortest path from $w$ to $t$ is $w-u_1-u_2-t$. The shortest paths from $s$ to $w$ and $w$ to $t$ trivially overlap, since they share the edge $u_1-u_2$. 

\begin{figure}[htbp]
	\begin{centering}
		\textsf{\includegraphics[width=0.8\linewidth]{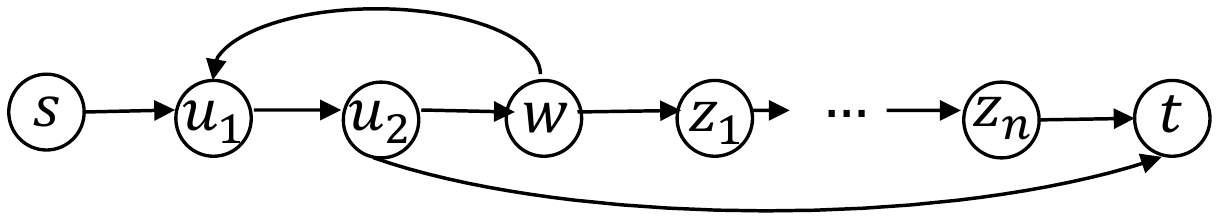}}
		\par\end{centering}
	
	\caption{Cycle in shortest-path based segment routing.}
	\label{cycle}
\end{figure}

In that case a natural question arises:
what if we consider node-constrained TE with shortest paths, under the condition that the resulting walk from the source to the destination is a path or even a simple path? As our subsequent analysis shows, traffic engineering generally becomes NP-hard, even for just one commodity, and even in the special case where the traffic must use all middlepoints in the input order. To prove this fact, we first introduce the following fundamental result due to Eilam-Tzoreff \cite{t1998}.

\begin{theorem}[NP-hardness of $k$DSP \cite{t1998}]
\label{theorem:kDSP}
Given a graph $G=(V,E)$ and $k$ pairs of distinct vertices $(u_i,v_i)$, $1\leq i\leq k$, the $k$DSP problem of computing $k$ pairwise disjoint shortest paths $P_i$ between $u_i$ and $v_i$ is NP-complete, when $k$ is part of the input. This result holds for all four versions of the $k$DSP problem, namely, node or edge-disjoint paths for directed or undirected graphs.
\end{theorem}

\begin{proposition}
\label{prop:SRshortestpathNP-hard}
The $TE_{MF}$ and $TE_{LU}$ problems in a directed graph with $K$ middlepoints $(1)$ using only shortest paths and $(2)$ only allowing paths or simple paths from a source to a destination that use all middlepoints are NP-hard, even for just one commodity, when $K$ is part of the input.
\end{proposition}

\begin{proof}[Proof sketch]
We can show the statement by making similar arguments as for general node-constrained TE in \cref{sec:theory}. For $TE_{MF}$, the idea is to first show that we can solve $k$DSP if and only if we can solve the corresponding $TE_{MF}$ formulation. 

Indeed, assume a directed graph $G=(V,E)$, two distinct nodes $s,t$ in $V$, and $K$ distinct nodes $s\neq M_i\neq t$ in $V$, $1\leq i\leq K$. Consider we do node-constrained routing from $s$ to $t$ using nodes $M_i$ as our $K$ middlepoints. We will show that the $TE_{MF}$ problem is NP-hard by a reduction from the $k$DSP problem.

Concretely, assume for instance the $k$DSP node-disjoint problem in Theorem \ref{theorem:kDSP}. We construct a new graph $G'$ as follows. For each $i$, $1\leq i\leq K-1$, we introduce a new node $M_i$ along with the two directed edges $e^i_{in}=(v_i,M_i)$ and $e^i_{out}=(M_i, u_{i+1})$. Moreover, we associate each edge in $G'$ with a unit capacity, and we assume the single commodity $(s,t)$ with source $s=u_1$ and destination $t=v_K$. We now argue that there are $K$ node-disjoint shortest paths between $u_i$ and $v_i$, if and only if $TE_{MF}$ with the single commodity $(s,t)$ and the $K-1$ middlepoints $M_1,\dots,M_{K-1}$ accepts a positive solution.

This can be proven using similar techniques as in \cref{sec:SRhardness}. The only difference is that an edge-cut now corresponds to a set of edges whose removal results in no path (or simple path) using shortest paths from the source to the destination through the middlepoints.
NP-hardness for $TE_{LU}$ follows immediately by Lemma \ref{lemma:connection}, in a similar spirit as Corollary \ref{corollary:NP-hard-TE-LU}.
\end{proof}





Proposition \ref{prop:SRshortestpathNP-hard}  assumes that the number of middlepoints $K$ is part of the input since in general we can have up to $|V|$ middlepoints. What about the case when $k$ is fixed? When $k=2$, \cite{t1998} provides a polynomial algorithm for the undirected case of $k$DSP, whereas the complexity for the directed case when $k=2$ remains open\footnote{The problem appears to have been answered in the affirmative in \cite{berczi}, when the length of each
edge is positive.}. On the other hand, only few partial results are known when $k$ is fixed and greater than 2 \cite{berczi}.

%
\section{Application to Flow Centrality}
\label{sec:flowcentrality}
Having studied various node-constrained TE formulations, in this section we move to another fundamental and practical question: how to actually select the middlepoints that lead to good TE performance in the first place?
We draw upon a concept called flow centrality from graph theory as an intuitive solution approach. 
Flow centrality, first introduced by Freeman et al. \cite{fbw1991}, characterizes a node's significance in terms of the maximum flow that can go through that node in the underlying flow network. Flow centrality is closely related to node-constrained TE for two reasons. First, it involves by definition node-constrained traffic. Second, flow centrality and its group extension answer the question of how to select the middlepoints in the general node-constrained TE (\cref{sec:theory}), in order to maximize the traffic through these points. Note that flow centrality is naturally related to the first TE type $TE_{MF}$, since it cares about the maximum flow objective.



The rest of this section is organized as follows. We provide two flow centrality definitions for single- and multi-commodity flow networks and discuss their computational complexity in \cref{sec:flow-centrality}. 
Furthermore, we introduce and analyze the group flow centrality and $N$-group maximum flow in \cref{sec:group-centrality}, and show that unlike other common graph centralities it does not fall under the framework of submodular function maximization.

\subsection{Flow Centrality}
\label{sec:flow-centrality}
The original flow centrality \cite{fbw1991} of a node $w\in V$ in a flow network $G=(V,E,c)$ is:

\begin{align}
\label{eq:trad-flow-centr}
\gamma(w)=\frac{\mathlarger{\sum}\limits_{s',t'\in V-\{w\}|s'\neq t'}\nu_{max}^w(s',t')}{\mathlarger{\sum}\limits_{s',t'\in V-\{w\}|s'\neq t'}\nu_{max}(s',t')}, 
\end{align}
where $\nu_{max}(s',t')$ is the maximum flow in the single-commodity flow network with commodity $(s',t')$, and $\nu_{max}^w(s',t')$ is the maximum flow through node $w$ in the single-commodity flow network with commodity $(s',t')$. Thus, the flow centrality of node $w$ represents the percentage of the maximum flow that can go through $w$ for a demand chosen uniformly at random.

For multi-commodity networks with given commodities, we introduce an alternative definition: 
\begin{align}
\label{eq:new-flow-centr}
\widetilde{\gamma}(w)=\frac{\nu_{max}^w(\bm{s},\bm{t})}{\nu_{max}(\bm{s},\bm{t})},
\end{align}
which denotes the ratio of the maximum multi-commodity flow $\nu_{max}^w(\bm{s},\bm{t})$ that can go through node $w$ to the maximum multi-commodity flow, assuming we are given $(\bm{s},\bm{t})$.

The basic difference in the two definitions is that the former considers equiprobably all possible source-destination pairs, while the latter focuses on the actual commodities in the flow network. Thus, the former is based on the single-commodity formulation, and the latter on the multi-commodity one.


For directed graphs, we show:

\begin{proposition}
\label{prop:directed-flow}
Given a flow network $G=(V,E,c)$ with directed edges and a node $w\in V$, it is NP-hard to compute $\gamma(w)$ or $\widetilde{\gamma}(w)$.
\end{proposition}
\begin{proof}
For the first statement, our strategy will be to show that
computing $\gamma(w)$ cannot be less hard than computing $\nu_{max}^w(s,t)$, for any source-destination pair $(s,t)$ with $s\neq t$ and $s,t\neq w$. This, in turn, establishes NP-hardness for $\gamma(w)$.
For convenience, we introduce the shorthand notation $\mathcal{S}^w(G)=\sum\limits_{s',t'\in V-\{w\}|s'\neq t'}\nu_{max}^w(s',t')$ to represent the numerator for $G$. 

In this direction, we consider the new flow network $G_{\hat{s}}$, which we construct from $G$ by introducing a new node $\hat{s}$ and a directed edge from $\hat{s}$ to $s$ with capacity equal to the sum of capacities of outgoing edges from $s$, i.e. $c((\hat{s},s))=\sum\limits_{e=(s,v)\in E}c(e)$. We now consider the quantity $\mathcal{S}^w(G_{\hat{s}})$ for the new flow network $G_{\hat{s}}$. $G_{\hat{s}}$ contains the same source-destination pairs as $G$, plus pairs of the form $(\hat{s},t'),\forall t'\in V$. The reason is that node $\hat{s}$ can only act as a source but not a destination for nodes in $G$. As a result, we get:
\begin{equation}\label{eq:flow-centr-1}
\mathcal{S}^w(G_{\hat{s}})=\mathcal{S}^w(G)+\sum\limits_{t'\in V}\nu_{max}^w(\hat{s},t').
\end{equation}
Next, we consider in a similar fashion the flow network $G^{\hat{t}}$, which we construct from $G$ by introducing a new node $\hat{t}$ and a directed edge from  $t$ to $\hat{t}$ with capacity equal to to the sum of capacities of incoming edges to $t$, i.e. $c((t,\hat{t}))=\sum\limits_{e=(v,t)\in E}c(e)$. We then proceed as before to show:
\begin{equation}\label{eq:flow-centr-2}
\mathcal{S}^w(G^{\hat{t}})=\mathcal{S}^w(G)+\sum\limits_{s'\in V}\nu_{max}^w(s',\hat{t}).
\end{equation}

Finally, we consider the flow network $G_{\hat{s}}^{\hat{t}}$, which we construct from $G$ by applying both steps described above for $G_{\hat{s}}$ and $G^{\hat{t}}$. The various graph constructions are illustrated in Figure \ref{figcentrality}. Let's now focus on $\mathcal{S}^w(G_{\hat{s}}^{\hat{t}})$. This quantity will contain the source destination pairs (i) within $G$, (ii) from $\hat{s}$ to nodes in $G$, (iii) from $\hat{t}$ to nodes in $G$, and (iv) from $\hat{s}$ to $\hat{t}$. But (iv) is precisely the same as from $s$ to $t$ due to the value of capacities of the new edges we introduced. Thus, we can write:
\begin{align}\label{eq:flow-centr-3}
\mathcal{S}^w(G_{\hat{s}}^{\hat{t}})&=\mathcal{S}^w(G)+\sum\limits_{t'\in V}\nu_{max}^w(\hat{s},t')\notag\\
&+\sum\limits_{s'\in V}\nu_{max}^w(s',\hat{t})+\nu_{max}^w(s,t).
\end{align}
Combining \eqref{eq:flow-centr-1}, \eqref{eq:flow-centr-2}, \eqref{eq:flow-centr-3}, we get:
\begin{equation}\label{eq:flow-centr-4}
\nu_{max}^w(s,t)=\mathcal{S}^w(G_{\hat{s}}^{\hat{t}})-\mathcal{S}^w(G_{\hat{s}})-\mathcal{S}^w(G^{\hat{t}})+\mathcal{S}^w(G).
\end{equation}
\eqref{eq:flow-centr-4} suggests that computing $\gamma(w)$ (and in turn $\mathcal{S}^w$) cannot be less hard than computing $\nu_{max}^w(s,t)$, since being able to compute the former implies that we can easily compute the latter via \eqref{eq:flow-centr-4}, which only involves polynomial-time constructions. Given it is NP-hard to compute $\nu_{max}^w$ by Proposition \ref{prop:NPhardness}, computing $\gamma(w)$ is also NP-hard.
\begin{figure}[t]
	\begin{centering}
		\textsf{\includegraphics[width=0.85\linewidth]{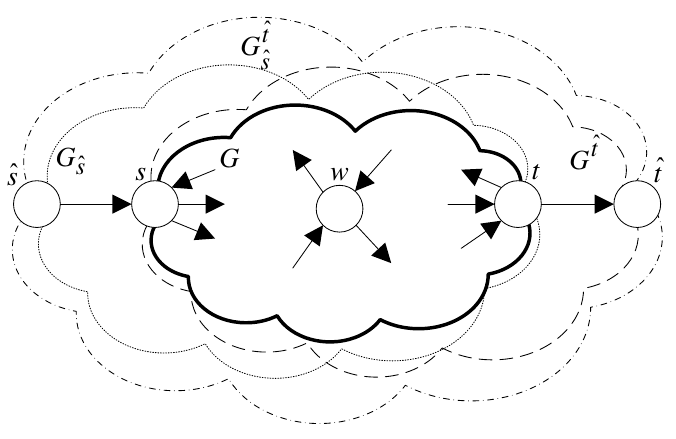}}
		\par\end{centering}
	
	\caption{Illustration of Proposition \ref{prop:directed-flow}.}
	\label{figcentrality}
\end{figure}

For the second statement, note again that $\widetilde{\gamma}(w)$ is a fraction of two terms. The denominator can be computed in strongly polynomial time \cite{Tardos1986}. If $\widetilde{\gamma}(w)$ were polynomially computable, then we could compute $\nu_{max}^w$ in strongly polynomial time as the product of $\widetilde{\gamma}(w)$ and $\nu_{max}$, which is a contradiction since by Proposition \ref{prop:NPhardness} it is NP-hard to compute $\nu_{max}^w$.
\end{proof}
For undirected graphs, we can similarly prove that:
\begin{corollary}
\label{cor:undirected-flow}
Given a flow network $G=(V,E,c)$ with undirected edges and a node $w\in V$, we can compute the flow centrality $\gamma(w)$ in Equation \eqref{eq:trad-flow-centr} or the flow centrality $\widetilde{\gamma}(w)$ in Equation \eqref{eq:new-flow-centr} in strongly polynomial time.
\end{corollary}
\begin{proof}
It follows immediately from the fact that $\nu_{max}^w(s,t)$ is strongly polynomially computable in undirected graphs.
\end{proof}

\subsection{Group Flow Centrality and $N$-Group Maximum Flow}
\label{sec:group-centrality}

In this section, we introduce the concept of multi-commodity group flow centrality, which is a generalization of the multi-commodity flow centrality $\widetilde{\gamma}$ through node $w$ in Equation~\eqref{eq:new-flow-centr} to a group of nodes $C$. In this work, we focus on the group extension of $\widetilde{\gamma}$ since in many practical applications we are given the commodities and must decide on a good set of nodes to select as middlepoints. 
Similar definitions exist in prior work on graph centrality; for instance, group betweenness centrality \cite{eb1999} naturally generalizes betweenness centrality \cite{F1977} to a group of nodes. 
Group flow centrality is based on the concept of group maximum flow.

\begin{definition}
The group multi-commodity maximum flow $\mathcal{GF}:2^V\rightarrow\mathbb{R}_{\geq 0}$ in a 
multi-commodity 
flow network with  commodities $(\bm{s},\bm{t})$ is a function which, for any group of nodes $C\subseteq V$, returns the maximum multi-commodity flow  $\nu_{max}^{C}(\bm{s},\bm{t})$ that can go through \textit{any} node in $C$, i.e. $\mathcal{GF}(C)=\nu_{max}^{C}(\bm{s},\bm{t})$.
\end{definition}

\begin{definition}
The group flow centrality of a group $C\subseteq V$ in a multi-commodity flow network is defined as:
\begin{equation}\label{equation:group-flow-centrality}
\tilde{\gamma}(C)=\frac{\nu_{max}^{C}(\bm{s},\bm{t})}{\nu_{max}(\bm{s},\bm{t})}.
\end{equation}
\end{definition}
The group flow centrality represents the percentage of the maximum multi-commodity flow that goes through any node in group $C$ for a given set of demands. It is obviously NP-hard for directed graphs as a generalization of the $\widetilde{\gamma}$ flow centrality which is NP-hard by Corollary~\ref{cor:undirected-flow}. But then the group flow centrality is also NP-hard for directed graphs by Equation \eqref{equation:group-flow-centrality}.

An important question concerns the selection of a group of nodes of at most a given size which achieves the largest possible maximum flow, or, equivalently, maximizes the group flow centrality.  For this purpose, we introduce the $N$-\textit{group maximum multi-commodity flow} with $N\in\mathbb{N}^{+}$:
\begin{equation}
\mathcal{GF}^N=\max_{C\subseteq V:|C|\leq N}\nu_{max}^{C}(\bm{s},\bm{t}).
\end{equation}

For single-commodity networks, the question is trivial since we can always put the source into the set $C$, but the problem turns out to be NP-hard in directed multi-commodity networks.
\begin{proposition}
\label{group-NPhard}
The $N$-group maximum multi-commodity flow is NP-hard for directed graphs.
\end{proposition}
\begin{proof}
We prove NP-hardness by reduction from the maximum coverage problem (MCP) \cite{nwf1978}, which is NP-hard. 
Assume a set of $m$ items $I=\{i_1,\dots,i_m\}$ and a collection of $n$ sets $S=\{S_1,\dots,S_n\}$, where each $S_i$ contains elements from $I$. Given a positive integer $N\leq n$, the MCP tries to select a subset of $S'\subseteq S$ of cardinality $|S'|\leq N$ such that the maximum number of elements are covered, i.e. the union of the selected sets has maximal size.
We can reduce the MCP to the $N$-group maximum single-commodity flow by constructing a directed graph $G=(\mathcal{M},E)$ as follows. 
$\mathcal{M}$ contains a pair of nodes $z_j$ and $u_j$ for each item $i_j$ that appears in $I$, and one node $v_k$ for each set $S_k$. We denote the set of nodes $z_j$ as $Z$, the set of nodes $u_j$ as $U$, and the set of nodes $v_k$ as $V$, so that $\mathcal{M}=Z\cup U\cup V$. We introduce $m$ edges of the form $(z_j,u_j)$, $j=\{1,\dots,m\}$. We add edge $(u_j,v_k)$, if and only if set $S_k$ contains item $i_j$. All edges have a capacity of 1. Finally, for each pair $(i_j,v_k)$ with the property that set $S_k$ contains item $i_j$, we consider a commodity $(z_j,v_k)$ with source $z_j$ and destination $v_k$. Thus, we add commodity $(z_j,v_k)$, if and only if the graph contains edge $(u_j,v_k)$. All commodities have infinite demand.
The polynomial construction is depicted in Figure \ref{fig:3}.

\begin{figure}[h]
	\begin{centering}
		\textsf{\includegraphics[width=0.9\linewidth]{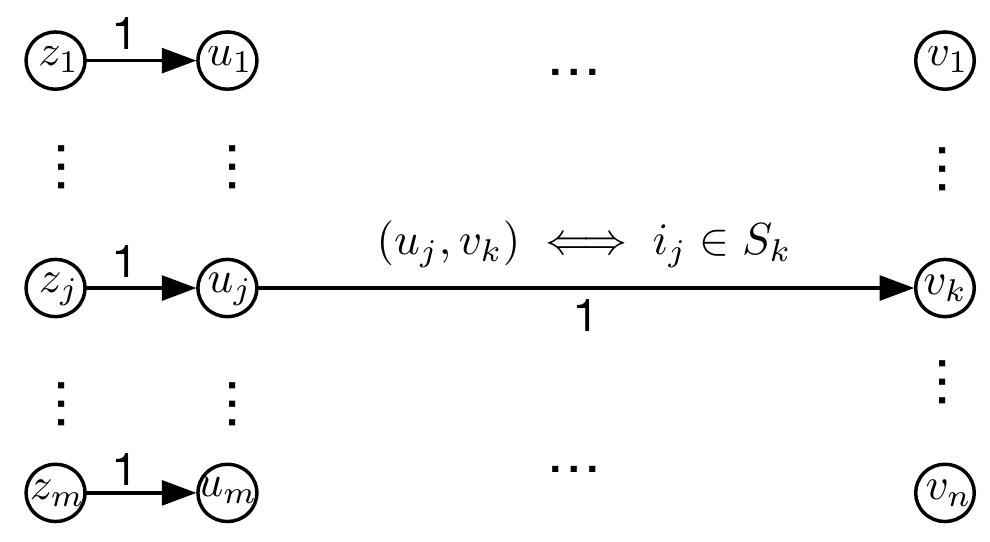}}
		\par\end{centering}
	\caption{Illustration of Proposition \ref{group-NPhard}.}
	\label{fig:3}
\end{figure}

We now prove that the maximum coverage problem has value $C_{max}$, if and only if the $N$-group maximum flow has a value of $C_{max}$. Assume first that the MCP has a value of $C_{max}$. We can then construct a corresponding flow in $G$ as follows. 
For each item $i_j$ that is covered in the optimal MCP solution, we randomly pick one set $S_k$ that covers $i_j$ in the solution  (there may be more than one sets covering $i_j$). We then send one unit of flow for commodity $(z_j,v_k)$ from $z_j$ to node $u_j$, and subsequently from $u_j$ to $v_k$.
The constructed multi-commodity flow is valid since it is easy to verify that it respects all capacity and conservation constraints. Based on that flow, we can also form a $N$-group flow. Indeed, by definition the MCP contains at most $N$ sets $S_i$, so there can be at most $N$ nodes of the form $v_k$ participating in the flow. Since the entire flow has to pass through these nodes, we claim that the above flow is a $N$-group flow passing through (at most) $N$ nodes of the form $v_k$. So, the $N$-group multi-commodity flow is at least $C_{max}$.

For the reverse direction, we argue that the $N$-group maximum flow cannot be greater than $C_{max}$. To show this, assume a group $\mathcal{G}$ of at most $N$ nodes that can accept a flow of value $C_N>C_{max}$. 
The nodes in $\mathcal{G}$ can belong to either $Z$ or $U$ or $V$. However, note that the flow through node $z_j$ (for any commodity of the form $(z_j,v_k)$) is equal to the flow through $u_j$. Hence, we can replace any node $z_j$ appearing in $\mathcal{G}$ by node $u_j$, without affecting the value of the group flow.
Let's thus write $\mathcal{G}=U_N\cup V_N$, where $U_N\subseteq U$ and $V_N\subseteq V$. We also define $U'\subseteq U-U_N$ to be the subset of nodes in $U-U_N$ that transmit some positive (non-zero) flow to any node in $V_N$.
We next argue that for any node $u_j\in U_N$, the total flow through $u_j$ over all commodities  with source $z_j$ must be one. If that were not the case, we could form a $N$-group flow of higher value by increasing the flow for any commodity $(z_j,v_k)$ until the total flow on  edge $(z_j,u_j)$ becomes one. This is possible since any commodity of the form $(z_j,v_k)$ can only route its flow through $u_j$.
This process respects all capacity constraints and achieves a higher group flow, which is a contradiction since we assumed a maximum $N$-group flow.
Similarly, we argue that the total flow over all commodities from any node $u_j\in U'$ to the set of nodes in $V_N$ must be one; otherwise, we could get a maximum $N$-group flow of higher value by first removing any flow (if any) through $u_j$ to nodes in $V-V_N$, and subsequently
increasing the flow for any commodity $(z_j,v_k)$ where $v_k\in V_N$ until the total flow through node $u_j$ becomes 1. This is always possible and respects all capacity constraints.
The two observations imply that in the maximum $N$-group flow, the total flow over all commodities to nodes in $U_N$ or $U'$ must be one, and as a result $C_N$ is an integer, equal to the cardinality $|U_N|+|U'|$.
Based on the $N$-group maximum flow, we can form a coverage for the original problem as follows. 
For each $u_j\in U_N$, we pick any node $v_k\in V$ where $u_j$ transmits some positive flow (at least one such node must exist); similarly, for each $u_j\in U'$, we pick any node $v_k\in V_N$ where $u_j$ transmits some positive flow.
Based on that, each item $i_j$, where $u_j\in U_N$ or $u_j\in U'$, is assigned to set $S_k$ corresponding to the $v_k$ above. 
The constructed coverage has value $C_N$ since it uses all $C_N$ items in  $U_N\cup U'$. Moreover, it uses at most $N$ sets $S_k$, since the number of sets cannot exceed the sum $|U_N|+|V_N|\leq N$.
This is a contradiction, since we assumed that the maximum coverage has value $C_{max}<C_N$.
\end{proof}


An interesting corollary of Proposition \ref{group-NPhard} is the following.

\begin{corollary}
\label{cor:approximation}
The $N$-group maximum flow is not possible to approximate in directed graphs within $1-\frac{1}{e}+o(1)$, unless $P=NP$.
\end{corollary}
\begin{proof}
From Proposition \ref{group-NPhard} we can compute the MCP by solving an $N$-group multi-commodity maximum flow problem. In particular, the MCP has a value equal to $C_{max}$, if and only if the $N$-group maximum flow has the same value $C_{max}$. If we were able to approximate the latter problem within $1-\frac{1}{e}+o(1)$, that would imply that we can use the above construction to approximate the MCP within a factor of $1-\frac{1}{e}+o(1)$, which is however not possible unless P $=$ NP \cite{f1998}. 
\end{proof}


\begin{remark}
\label{remark-4}
In fact, Proposition \ref{group-NPhard} and Corollary \ref{cor:approximation} are also true for undirected graphs. We omit the full details here, but the main idea is as follows. We use the same construction as in Proposition \ref{group-NPhard}, albeit with undirected edges.
We can first show similar to the directed case that the $N$-group multi-commodity flow in the undirected graph has value at least as large as that of the MCP. For the reverse direction, we note that the group flow in the undirected graph can generally be at least as large as in the directed graph, since it can use any path in the directed graph (and potentially more). However,
we observe that it is not possible for the undirected graph that we construct to achieve a higher group flow than the directed one, since the directed graph already has edges of the form $(u_j,v_k)$ for all commodities $(z_j,v_k)$, which can directly route the traffic from $z_j$ to $v_k$. Thus, the maximum group flows in the undirected or directed networks are the same, and this implies that the $N$-group multi-commodity flow in the undirected graph has value at most as large as that of the MCP.
\end{remark}


A natural question is whether group multi-commodity flow falls under the paradigm of submodular function maximization. First, the MCP that we used in the reduction of Proposition \ref{prop:directed-flow} falls under this paradigm. Second, similar results already exist in graph centrality theory. For instance, group closeness centrality is shown to be NP-hard, monotone, and submodular \cite{chen2016}. This is also true for group betweenness centrality \cite{depz2009}.


\begin{definition}
Consider a finite set of elements $U$ and a function $g:2^U\rightarrow\mathbb{R}_{\geq 0}$. We call $g$ monotone if adding an element to any set $S$ 
 cannot cause the function to decrease, i.e., $g(S\cup\{v\})\geq g(S)$ for all $v\in U$ and $S\in 2^U$. Furthermore, we call $g$ submodular if the marginal gain from adding an element to any set $S$ is at least as high as the marginal gain from adding the same element to a superset of $S$, i.e., $g(S\cup\{v\})-g(S)\geq g(T\cup\{v\})-g(T)$ for all $v\in U-T$ and pairs of sets $S\subseteq T$.
\end{definition}
Interestingly, our next result shows that multi-commodity group flow is not submodular even though it is monotone.
\begin{lemma}
\label{monot_submod}
The function $\mathcal{GF}:2^{V}\rightarrow\mathbb{R}_{\geq 0}$ is monotone but not submodular in directed or undirected graphs.
\end{lemma}
\begin{proof}
For monotonicity, note that adding a node can never decrease the maximum group flow, since an additional node does not decrease the number of available paths for the flow (either it increases them or it leaves their number unchanged).
\begin{figure}[h]
	\begin{centering}	\textsf{\includegraphics[width=0.6\linewidth]{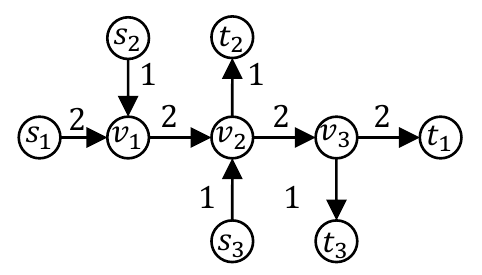}}
		\par\end{centering}
	\caption{Illustration of Lemma \ref{monot_submod}.}
	\label{fig:8}
\end{figure}
To show that $\mathcal{GF}$ is not submodular in directed graphs, it is sufficient to construct a proper counterexample. To this goal, consider the directed multi-commodity flow network of Figure \ref{fig:8} with three commodities: commodity $(s_1,t_1)$ with a demand of 2 units, and commodities $(s_2,t_2)$ and $(s_3,t_3)$, each with a unit demand. Assume the following two sets $S=\{s_1\}$ and $T=\{s_1,s_2\}$, with $S\subset T$. For the corresponding group multi-commodity flows it holds that $\mathcal{GF}(S)=\mathcal{GF}(T)=2$. The reason is that edge $(v_1,v_2)$ with a capacity of 2 acts as the bottleneck that limits the total network flow to a maximum of 2. Consider now node $s_3$. We can again argue that $\mathcal{GF}(S\cup\{s_3\})=2$, since edge $(v_2,v_3)$ with capacity 2 acts as the bottleneck that limits the total group flow to 2 units. However, note that $\mathcal{GF}(T\cup\{s_3\})=3$: commodity $(s_1,t_1)$ sends one unit of flow along $s_1\to v_1\to v_2\to v_3\to t_1$, commodity $(s_2,t_2)$ sends one unit along $s_2\to v_1\to v_2\to t_2$, and commodity $(s_3,t_3)$ sends one unit along $s_3\to v_2\to v_3\to t_3$. This suggests that $\mathcal{GF}(T\cup\{s_3\})-\mathcal{GF}(T)=3-2=1>\mathcal{GF}(S\cup\{s_3\})-\mathcal{GF}(S)=1-1=0$, which implies that the group multi-commodity flow $\mathcal{GF}$ is not a submodular function.
Finally, to show that $\mathcal{GF}$ is not submodular in undirected graphs, we can use exactly the same counterexample as above, except that the edges are undirected this time.
\end{proof}

One implication of Lemma \ref{monot_submod} is that we cannot use the standard greedy algorithm by Nemhauser et al. \cite{nwf1978} to find a $(1-\frac{1}{e})$-approximate solution to the group flow. We note nevertheless that the greedy algorithm would in any case not be practical for directed graphs, given that computing the maximum $w$-flow is already NP-hard by Proposition \ref{prop:NPhardness}.

\section{Related Work}
\label{sec:related}

\subsection{Segment Routing and TE}
\label{sec:related-te}
Segment routing \cite{segment1, segment2, segment3} is a recent paradigm that facilitates packet forwarding via a series of segments. 
Segment routing has been explored with TE. Bhatia et al.
\cite{bhatia2015optimized} apply 2-segment routing to TE, where any logical
path contains only one middlepoint and thus two segments. Hartert et al.
\cite{hartert2015solving,HVSB15} propose heuristics to solve various TE
problems with segment routing. 
Contrary to these works, our goal is 
to study the fundamentals of several variants of node-constrained TE for both directed and undirected graphs. 
Aubry et al. \cite{aubry2016scmon} propose to use
segment routing for continuous monitoring of the data plane of the network
with a single box. Segment routing is used to force probe packets to traverse
specific paths. Giorgetti et al. \cite{giorgetti2015path} propose algorithms
for segment routing label stack computation that guarantee minimum label stack
depth. Such use cases are beyond our work.

TE has been extensively studied in carrier networks 
\cite{EJLW01,KKDC05,WXQY06,FT00,HBCR07,HVSB15}, and recently in data center backbone WANs 
\cite{JKMO13,HKMZ13,LKMZ14,GHCR13} with software defined networking 
\cite{FRZ13}. In general it is assumed that TE can use any valid path in the network, or any path from a predetermined set of paths. Node-constrained TE is clearly different.

\subsection{Graph Centrality}
\label{rel:centrality}
The centrality concept from graph theory and network analysis
\cite{Newman2010} identifies the most important vertices in a graph.
Centrality was first developed in social network analysis \cite{F1977,b87} to
determine the most influential nodes.
We review two relevant centrality metrics.
\textit{Betweenness} centrality characterizes the power of a node in terms of the number of shortest paths that go through that node for a randomly picked source-destination pair. 
Brande's algorithm can compute this centrality in polynomial space and time
\cite{B2001}.
\textit{Closeness} centrality \cite{freeman1978} of a node is calculated as the sum of the length of the shortest paths between that node and all other nodes in the graph. 
As opposed to the aforementioned \textit{individual} centrality, the \textit{group} centrality of a group of nodes $C\subseteq V$ refers to the combined centrality of the group \cite{eb1999}. 
Group betweenness centrality can be approximated within a factor $1-\frac{1}{e}$ to the optimal 
\cite{depz2009,pyeb2009}. 
Group closeness centrality can also be approximated within a factor of $1-\frac{1}{e}$ using the standard greedy algorithm 
\cite{chen2016}.


Note that graph centralities have been applied to routing in some SDN problems, such as service chain embedding \cite{Lukovszki2015}
 and incremental SDN deployment 
\cite{Lukovszki2016,Levin2014}. 
Solutions to these problems are based on degree centralities, and use greedy approximation algorithms exploiting submodularity \cite{Lukovszki2016}. 

\section{Conclusion}
\label{sec:conclusion}

In this work we study the fundamentals of node-constrained TE, where the traffic is
constrained to go through specific middlepoints. We show that the general
node-constrained TE problem is NP-hard for directed graphs, but strongly
polynomial for undirected graphs. Furthermore, node-constrained TE with
shortest paths is weakly polynomial, but its acyclic variant is generally NP-hard. 
An application of node-constrained TE concerns flow centrality, whose computational complexity we derive for both directed and undirected graphs. Lastly we introduce and study group multi-commodity flow centrality.
 
Our work is important because node-constrained TE has wide applicability in
emerging networking technologies such as segment routing but also because of its direct connection to the flow centrality concept.
Our hardness results
hint at the practical limitations of many variants of node-constrained TE. For
this reason, an important direction for future research is the development of practical algorithms for the various computationally hard TE variants as well as the group multi-commodity flow. 



\section*{Acknowledgment}
We thank the anonymous reviewers whose invaluable comments and suggestions helped improve and clarify this manuscript.


\bibliographystyle{IEEEtranS}
\balance
\bibliography{IEEEabrv,bib/reference}

\end{document}